\documentclass[11pt]{article} 
\usepackage[dvipsnames,svgnames,x11names]{xcolor}
\usepackage[utf8]{inputenc}
\usepackage{authblk}
\usepackage{fullpage}
\usepackage{tabularx}
\usepackage{amsthm}
\usepackage{amsmath}
\usepackage{amssymb}
\usepackage{amsfonts}
\usepackage{mathtools}
\usepackage{enumitem}
\usepackage[pagebackref]{hyperref}
\usepackage[sort,numbers]{natbib}
\usepackage{tikz}
\usepackage[capitalize]{cleveref}
\usepackage{todonotes}

\usepackage{multirow}
\usepackage{booktabs}
\usepackage{makecell}
\usepackage{caption}

\newtheorem{theorem}{Theorem}
\newtheorem{lemma}[theorem]{Lemma}
\newtheorem{observation}[theorem]{Observation}
\newtheorem{corollary}[theorem]{Corollary}
\newtheorem{proposition}[theorem]{Proposition}
\newtheorem{remark}{Remark}
\newtheorem{claim}{Claim}
\newtheorem{definition}{Definition}

\theoremstyle{definition}
\newtheorem{construction}{Construction}
\usepackage{soul}

\crefname{figure}{Figure}{Figures}

\newcommand{\commentout}[1]{}

\newenvironment{claimproof}{\emph{Proof of Claim.}}{\hfill$\diamond$
}

\newcommand{\problemdef}[3]{
	\begin{center}\fbox{
	\begin{minipage}{0.95\textwidth}
		\noindent
		#1
		\vspace{5pt}\\
		\setlength{\tabcolsep}{3pt}
		\begin{tabularx}{\textwidth}{@{}lX@{}}
			\textrm{Input:}     & #2 \\
			\textrm{Question:}  & #3
		\end{tabularx}
	\end{minipage}}
	\end{center}
}

\newcommand{\problemdefopt}[3]{
	\begin{center}\fbox{
	\begin{minipage}{0.95\textwidth}
		\noindent
		#1
		\vspace{5pt}\\
		\setlength{\tabcolsep}{3pt}
		\begin{tabularx}{\textwidth}{@{}lX@{}}
			\textrm{Input:}     & #2 \\
			\textrm{Task:}  & #3
		\end{tabularx}
	\end{minipage}}
	\end{center}
}

\newcommand{\OO}{\mathcal{O}}

\newcommand{\probname}{\textsc{Tournament Value Maximization}}

\begin{document}
\title{How to Make Knockout Tournaments More Popular?}

\author{Juhi~Chaudhary\thanks{Supported by the European Research Council (ERC) project titled PARAPATH (101039913).}}
\author{Hendrik~Molter\thanks{Supported by the ISF, grant No.~1456/18, and the ERC, grant number 949707.}}
\author{Meirav~Zehavi$^*$}

\affil{\small Department of Computer Science, Ben-Gurion~University~of~the~Negev, 
Beer-Sheva, 
Israel\\ \texttt{juhic@post.bgu.ac.il, molterh@post.bgu.ac.il, meiravze@bgu.ac.il}}

\date{}
\maketitle
\begin{abstract}
Given a mapping from a set of players to the leaves of a complete binary tree (called a {\em seeding}), a {\em knockout tournament} is conducted as follows: every round, every two players with a common parent compete against each other, and the winner is promoted to the common parent; then, the leaves are deleted. When only one player remains, it is declared the winner.  This is a popular competition format in sports, elections, and decision-making. Over the past decade, it has been studied intensively from both theoretical and practical points of view. Most frequently, the objective is to seed the tournament in a way that ``assists'' (or even guarantees) some particular player to win the competition. We introduce a new objective, which is very sensible from the perspective of the directors of the competition:  maximize the profit or popularity of the tournament. Specifically, we associate a ``score'' with every possible match, and aim to seed the tournament to maximize the sum of the scores of the matches that take place. We focus on the case where we assume a total order on the players' strengths, and provide a wide spectrum of results on the computational complexity of the problem.
\end{abstract}

\section{Introduction}\label{sec:intro}
A {\em knockout (or single-elimination) tournament} is the most popular competition format in sports~\cite{DBLP:journals/ior/HorenR85,CR11,GMSS12}. Here, roughly speaking, the players are paired up to compete against each other in rounds, where, at each round, the losers are knocked out, until only one player (the winner) remains. For an illustrative example, consider any major tennis tournament or the knockout stage of the football World Cup. Notably, knockout tournaments are common in other fields as well, such as elections and decision-making~\cite{DBLP:conf/ijcai/Suksompong21,Tullock80,Rosen86,Laslier97}.

More formally, we are given $n$ \emph{players} (where, for simplicity, $n$ is a power of $2$) and a {\em seeding} that determines how to label the $n$ leaves of a complete binary tree with the players. Given a seeding, the competition is conducted in rounds: As long as the tree has at least two leaves, every two players with a common parent in the tree play against each other, and the winner is promoted to the common parent; then, the leaves of the tree are deleted from it. Eventually, only one player remains, and this player is declared the winner. 

Over the past decade,  knockout tournaments have been studied intensively from both theoretical and practical points of view.  Most of the attention has been given to the objective of making a specific player win, particularly by selecting a seeding that is considered advantageous to this player, as well as bribing some of the players (see, e.g., the surveys by \citet{DBLP:conf/ijcai/Suksompong21} and \citet{williams_moulin_2016}). Here, to decide whether a seeding is advantageous, we assume to possess {\em predictions} (deterministic or probabilistic) on the winners of (all or some of) the possible games or matches (we use these two terms interchangeably). In particular, when the aforementioned information is deterministic and complete (i.e., for every possible match involving two players, we ``know'' who will be the winner) and our only influence on the competition is by picking the seeding, the problem is known as the {\sc Tournament Fixing} problem. This problem and its various variations have been extensively studied from the viewpoints of classical complexity, parameterized complexity, and structural analysis~(see, e.g., \cite{DBLP:conf/ijcai/Gupta0SZ19,DBLP:conf/ijcai/GuptaR0Z18,DBLP:conf/ijcai/GuptaR0Z18a,DBLP:conf/ijcai/KimW15,DBLP:conf/aaai/AzizGMMSW14,DBLP:conf/atal/VuAS09,DBLP:conf/aaai/Williams10,zehavi2023tournament,RamanujanS17,DBLP:conf/ijcai/StantonW11,DBLP:conf/wine/StantonW11,DBLP:conf/aaai/KimSW16,manurangsi2023fixing}; this list is illustrative rather than comprehensive). 

One aspect of our contribution is conceptual: we introduce a new objective, which is, perhaps, more sensible from the perspective of the directors of the competition: maximize the profit or popularity of the tournament. Indeed, this is the main reason behind commercials and promotions for sports competitions and is a subject of active discussion in the media. More formally, the input consists of $n$ players and, for every possible match between two players, the winner of that match. Additionally, for every possible match between two players and the round in which it is to be conducted, we have the {\em value} of that match in that round, represented by an integer. Here, the value can correspond to the profit, popularity, or any other measure (or combination thereof) associated with the match in that round. We expect matches conducted in later rounds to be more profitable/popular; yet, when rounds do not affect this matter, the value function is said to be {\em round-oblivious}. 
The latter is plausible in many cases, such as local derbies or games between two rivaling teams or players, which roughly maintain their excitement regardless of the tournament round in which it occurs~\cite{hall2010empirical,chmait2020tennis}. 
The task is to compute a seeding that maximizes the sum of the values of the matches that take place in all rounds. Clearly, this problem can be naturally generalized to scenarios where the predictions are probabilistic or/and partial, to maximize the {\em expected} value.   

Being the first study on the computational complexity of this problem, we focus on the simplest and most basic case, where the predictions correspond to a linear ordering of the players, say, from weakest to strongest. This is also the most realistic case in terms of existing inputs. Indeed, a linear ordering can be simply obtained by taking an existing ranking of the players or calculating a number based on their previous performances (e.g., as is done in tennis~\cite{TennisRankingWomen,TennisRankingMen}). 
However, in many cases, it is unclear how to obtain more complicated, non-linear assessments of outcomes.
Moreover, intuitively, the deterministic model can serve as a ``reasonable approximation'' for settings where linear player rankings are available and the winning probabilities are either unknown or assumed to be close to 1/0. 
Henceforth, we refer to our problem---with a linear ordering of the players---as \probname; a formal definition can be found in \cref{sec:ps}. We note that for the {\sc Tournament Fixing} problem, the case of a linear ordering makes the problem trivial, since then, for any seeding, the same player (being the strongest one) wins. 

\paragraph*{Other Related Works.} There is a huge body of research on \textsc{Tournament Fixing} and related tournament manipulation problems; we gave an illustrative list in the introduction.
However, to the best of our knowledge, there is no prior work on \probname, with the following exception. \citet{dagaev2018competitive} investigated a highly restricted case of our setting, where every player has a distinct strength value, and the value of a game is determined by (a linear combination) of its ``quality'' (the sum of strengths of the involved players) and its ``intensity'' (the absolute difference of the strengths of the involved players). They characterize cases where either a ``close'' seeding is optimal, a ``distant'' seeding is optimal, or every seeding is optimal. In particular, this implies that their restricted cases are trivially solvable in linear time. 

\newcommand{\mrrb}[2]{\multirow{#1}{2.4cm}{\centering #2}}
\renewcommand{\arraystretch}{1.5}
\newcolumntype{Y}{>{\centering\arraybackslash}X}
\renewcommand\tabularxcolumn[1]{m{#1}}
\captionsetup[table]{skip=6pt}
\begin{table}[t]
  \setlength{\tabcolsep}{8pt}
  \centering
  \caption
  {Table of Results.}
  \def\fbs{3.5cm}
\small
  \begin{tabularx}{.95\textwidth}{@{}Y|Y|Y@{}}%
 \hline
      \textbf{Restriction} & \textbf{Hardness Results} & \textbf{Algorithmic Results}	\\
      \hline
      unrestricted & \text{NP}-hard and \text{APX}-hard for 2~game values (\cref{thm:apxhard}) & - \\
       \hline
       \multirow{2}{*}{\centering round-oblivious} & \multirow{2}{5cm}{\centering \text{NP}-hard and \text{APX}-hard for 3~game values (\cref{thm:apx2})} & $(1/\log n)$-factor approximation (\cref{thm:approx})\\\cline{3-3}
       & & FPT w.r.t.\ the size of a minimum influential set of players (\cref{thm:fptinfluence})\\

       \hline
win-count oriented & - & $n^{\OO(\log n)}$-time algorithm (\cref{thm:quasipoly})\\
  \hline
  \multirow{2}{4.5cm}{\centering player popularity-based (implies round oblivious and win-count oriented)} & - & linear-time algorithm for 2 player values (\cref{thm:twovalues}) \\\cline{3-3}
   &  & FPT w.r.t.\ the disagreement between the player popularity values and the strength ordering (\cref{thm:fptdisagree}) \\
  \hline
  \end{tabularx}
  \label{tab:results}
\end{table}

\paragraph*{Our Contributions.} 
We introduce \probname\ and analyze its computational complexity. 
In \cref{sec:apx}, we prove that it is \text{NP}-hard (as well as \text{APX}-hard) in two highly restricted scenarios: when all game values are $0$ or $1$, or when the game-value function is round-oblivious and there are 3 distinct game values. Here, the proofs are based on non-trivial reductions from {\sc Max $(2,3)$-Sat}.

Nevertheless, we provide a wide spectrum of positive results in \cref{sec:algos}. First, we provide a simple $(1/\log n)$-factor approximation algorithm based on the computation of a maximum-weight matching in a graph.\footnote{Throughout this document, $\log$ refers to the base-2 logarithm.} Second, we identify a large family of game-value functions that give rise to efficient algorithms: a quasipolynomial-time algorithm (i.e., with runtime $n^{\OO(\log n)}$). Intuitively, this family consists of game-value functions that allow the total value of a tournament to be computed from the information on the number of wins of each player. Here, our main tool is the introduction of the concept of {\em open} and {\em closed subtournaments} (given a partial seeding). We use it to design a dynamic programming algorithm. Moreover, we identify a natural restriction that facilitates a simple greedy algorithm. Here, each player is assigned a popularity value, and the value of a game is equal to the popularity value of the winning player. Additionally, we assume that there are only two different player popularity values. Intuitively speaking, players can be categorized as either ``popular" or ``unpopular".

Still regarding positive results, now at the parameterized complexity front, our contribution is twofold. First, we consider the setting where each player is once again assigned a popularity value, and the value of a game is equal to the popularity value of the winning player. 
Note that the popularity values naturally define an ordering of the players, from most to least popular.
For this setting, we consider a natural distance measure between the strength ordering and the ordering given by the popularity values of the players (for a formal definition, see \cref{sec:greedy}).
For this setting, we present a {\em fixed-parameter algorithm} (i.e., an algorithm with a runtime of the form $f(k)\cdot n^{\OO(1)}$ for a function $f$ depending only on $k$). The motivation for considering this parameter stems from the common observation that the popularity and fascination surrounding players are often highly correlated to their capabilities.
Second, for (general) round-oblivious value functions, we consider the parameter $k$ as the minimum size of a so-called {\em influential set of players}. Roughly speaking, we define an influential set as a set of players such that every match that does not involve any of them has value $0$ (for a formal definition, see \cref{sec:fptvc}). We design a fixed-parameter algorithm for this parameter as well. To this end, we adapt the recent algorithm of \citet{zehavi2023tournament} for {\sc Tournament Fixing} parameterized by the feedback vertex set number of the {\em prediction graph} of the input, being the complete digraph obtained by having an arc from a player $a$ to a player $b$ if $a$ is predicted to beat $b$. The motivation to consider this parameter is that, quite often, only a small set of players are truly profitable or popular.

Our main results are summarized in \cref{tab:results}. We hope that our work will open up the door for further studies of the maximization of the profit or popularity of tournaments, and present some directions for further research in \cref{sec:conc}.

\section{Problem Setting and Preliminaries}\label{sec:ps}
In our tournament setting, we are given a set $N$ of $n=|N|$ \emph{players} (where, for simplicity, we assume that $n$ is a power of $2$). Furthermore,  we assume to possess deterministic {\em predictions} on the winners of all possible matches. Generally, this is modeled by a so-called \emph{tournament graph}, a directed graph that has the set of players as vertices and where we have one arc between each pair of players. As mentioned in the introduction, our work focuses on the special case where we have a linear ordering on the players that defines their relative strength. We assume that the ordering is from strongest to weakest player. In other words, we consider the setting where the tournament graph is acyclic. This allows us to identify players with natural numbers and use the canonical ordering on the natural numbers as the strength ordering. 
%
%
More formally, we define \emph{players} as natural numbers, and we say that
a player $i$ \emph{beats} a players $j$ if $i>j$.

A {\em seeding} that determines how to label the $n$ leaves of a complete binary tree with the players. Given a seeding, the competition is conducted in rounds as follows. As long as the tree has at least two leaves, every two players with a common parent in the tree play against each other, and the winner is promoted to the common parent; then, the leaves of the tree are deleted from it. Eventually, only one player remains, and this player is declared the winner.

Formally, given a set $N\subseteq \mathbb{N}$ of players with $|N|=n=2^{n'}$ for some $n'\in\mathbb{N}$, we define 
a \emph{tournament seeding} for $N$ as an injective function $\sigma:N\rightarrow[n]$, where for $k\in \mathbb{N}$ we denote $[k]=\{1,2,\ldots,k\}$.
Informally speaking, $\sigma(i)$ is the seed position of player $i$ and corresponds to the ordinal position of the leaf player $i$ is assigned to in a DFS-ordering of the leaves of a complete binary tree with $n$ leaves.

For some $N'\subseteq N$ with $|N'|=2^{n''}$ and a seeding $\sigma$ for $N$, we denote the function $\sigma':N'\rightarrow [|N'|]$ as a \emph{partial seeding} for $N'$, if for all players $i,j\in N'$, we have that $\sigma'(i)-\sigma'(j)=\sigma(i)-\sigma(j)$.
We use a \emph{game-value function} $v:N\times N\times \mathbb{N}\rightarrow \mathbb{Z}$ to quantify the value of a game. If player $i$ plays against player $j$ in round $r$, the value of this game is $v(i,j,r)$. 
The value of the tournament is the sum of the values of the games played. Formally, we have the following.

\begin{definition}[Tournament Value]\label{def:tvalue}
Given a set of players $N$, a game-value function $v:N\times N\times \mathbb{N}\rightarrow \mathbb{Z}$, and a tournament seeding $\sigma$, the \emph{tournament value} $V_\sigma$ is defined as follows.
\begin{itemize}
    \item If $N=\{i,j\}$ with $i>j$, then for every tournament seeding $\sigma_0$ for players $i,j$ the tournament value $V_{\sigma_0}$ is $v(i,j,1)$ and player $i$ wins the tournament.
    \item Assume $|N|=2^n$ for some $n>1$. Let $N_1\subset N$ with $|N_1|=2^{n-1}$ and let $N_2=N\setminus N_1$ such that for all $i\in N_1$ and $j\in N_2$ we have that $\sigma(i)<\sigma(j)$. Let $\sigma_1$ denote the partial seeding for the players in $N_1$ and let $\sigma_2$ denote the partial seeding for the players in $N_2$.
    
    Let $V_{\sigma_1}$ be the value of the subtournament of players in $N_1$ and let $i_1$ be the winner of that subtournament. Analogously, let $V_{\sigma_2}$ be the value of the subtournament of players in $N_2$ and let $i_2$ be the winner of that subtournament. 
    
    Then, the value of the tournament is $V_\sigma=V_{\sigma_1}+V_{\sigma_2}+v(i_1, i_2, n)$ and if $i_1>i_2$, then $i_1$ wins the tournament, otherwise $i_2$ wins the tournament. 
\end{itemize}
\end{definition}

The main problem we investigate asks whether we can find a seeding that guarantees a certain minimal value of the tournament. Our formal problem definition is the following, where we assume that the set of players in the input is sorted by the strength of the players (from strongest to weakest).

\problemdef{\probname}{
A set $N\subset \mathbb{N}$ of players with $|N|=n=2^{n'}$ for some $n'\in\mathbb{N}$, a game-value function $v:N\times N\times \mathbb{N}\rightarrow \mathbb{Z}$, and a target value $V$.
}{
Is there a tournament seeding $\sigma$ for the players in $N$ such that the tournament value $V_\sigma$ is at least $V$?
}

We will consider some natural restrictions on the game-value function.
We call a game value $v$ function \emph{home team oblivious} if for all pairs of players $i,j$ and all rounds $r$ we have that $v(i,j,r)=v(j,i,r)$.
We call a game value $v$ function \emph{round-oblivious} if for all pairs of players $i,j$ and all rounds $r,r'$ we have that $v(i,j,r)=v(i,j,r')$.

Furthermore, we make the following observations.
\begin{observation}\label{obs:hometeam}
Let $I$ be an instance of \probname. Let $I'$ be the instance obtained from $I$ by replacing the game-value function $v$ with $v'$, where for all $i,j\in N$ and for all $r\in\mathbb{N}$ we have that $v'(i,j,r)=\max\{v(i,j,r),v(j,i,r)\}$. It holds that $I$ is a yes-instance if and only if $I'$ is a yes-instance.
\end{observation}
\begin{proof}
  $(\Leftarrow)$: Let $I$ be a yes-instance. Now, if we use the seeding used in $I$ in instance $I'$, then by the definition of the game-value function $v'$, the resulting tournament value of $I'$ is at least equal to that of $I$. 
  
  $(\Rightarrow)$: Let $I'$ be a yes-instance. Now, we construct a seeding $\sigma$ for $I$ in such a way that for every game played between players $i$ and $j$ in round $r$ of $I'$, the following hold: if $v'(i,j,r)=v(i,j,r)$, then $i$ is seeded before $j$ in $\sigma$, and if $v'(i,j,r)=v(j,i,r)$, then $j$ is seeded before $i$ in $\sigma$. If $v(i,j,r)=v(j,i,r)$, then the relative ordering of $i$ and $j$ in $\sigma$ is irrelevant. Now, note that the tournament value of $I$ computed using the seeding  $\sigma$ is equal to that of the tournament value of~$I'$.
\end{proof}

\cref{obs:hometeam} allows us w.l.o.g.\ to only consider \probname\ instances with home team oblivious game-value functions. Furthermore, we can observe that we can add some constant to each game value without significantly changing the tournament.

\begin{observation}\label{obs:constant}
Let $I$ be an instance of \probname. Let $I'$ be the instance obtained from $I$ by replacing the game-value function $v$ with $v'$, where for all $i,j\in N$ and for all $r\in\mathbb{N}$ we have that $v'(i,j,r)=v(i,j,r)+c$ for some $c\in\mathbb{Z}$ and by replacing the target value $V$ with $V'=V+(|N|-1)\cdot c$. It holds that $I$ is a yes-instance if and only if $I'$ is a yes-instance. 
\end{observation}
\begin{proof}
Considering the fact that exactly $|N|-1$ games occur in both tournaments $I$ and $I'$ and the game values in $I'$   are defined as $v'(i,j,r)=v(i,j,r)+c$, it follows that a seeding $\sigma$ used in tournament $I$ leads to a tournament value of $V$ if and only if the same seeding $\sigma$ results in a tournament value of $V+(|N|-1)\cdot c$ in tournament $I'$. Therefore, the observation is immediate.
\end{proof}



To find tractable cases, we investigate a restricted class of game-value functions that we call \emph{win-count oriented}. Intuitively speaking, such a game-value function allows one to evaluate the value of the tournament by looking at each player and the number of games they won.

\begin{definition}[Win-Count Oriented Game-Value Function]\label{def:wincount}
    A game-value function $v:N\times N\times \mathbb{N}\rightarrow \mathbb{Z}$ is \emph{win-count oriented} if there is a \emph{player evaluation function} $p:N\times\mathbb{N}\rightarrow \mathbb{Z}$ such that for every seeding $\sigma$ the tournament value $V_\sigma$ equals $\sum_{i\in N}p(i,w_{\sigma}(i))$, where $w_{\sigma}(i)$ denotes the number of wins of player $i\in N$ when tournament seeding $\sigma$ is used.
\end{definition}

We give an alternative characterization of win-count oriented game-value functions. We show that the game-value function is win-count oriented if and only if the function value only depends on the round and the winning player.

\begin{proposition}
A game-value function $v:N\times N\times \mathbb{N}\rightarrow \mathbb{Z}$ is win-count oriented if and only if there exists a function $v':N\times\mathbb{N}\rightarrow \mathbb{Z}$ such that for all $i,j,r\in N\times N\times \mathbb{N}$ we have
\[
v(i,j,r)=v'(\max(i,j),r).
\]
\end{proposition}
\begin{proof}
For a tournament seed~$\sigma$, let $\mathcal{G}_\sigma\subseteq N\times N\times \mathbb{N}$ denote the set of games played in the tournament.
If for all $i,j,r\in N\times N\times \mathbb{N}$ we have $v(i,j,r)=v'(\max(i,j),r)$, then we can set 
\[
p(i,w_\sigma(i))=\sum_{1\le r\le w_\sigma(i)}v'(i,r).
\]

Then, we have
\[
V_\sigma=\sum_{(i,j,r)\in\mathcal{G}_\sigma} v'(\max(i,j),r)=\sum_{i\in N} \sum_{(i,j,r)\in\mathcal{G}_\sigma \wedge i=\max(i,j)} v'(\max(i,j),r)=\sum_{i\in N} \sum_{1\le r\le w_\sigma(i)}v'(i,r)=\sum_{i\in N}p(i,w_\sigma(i)).
\]

It follows that $v$ is win-count oriented. In the remainder, we prove the converse direction.

Let $v: N\times N\times \mathbb{N}\rightarrow \mathbb{Z}$ be a win-count oriented game-value function. Then, we have
\[
V_\sigma=\sum_{(i,j,r)\in\mathcal{G}_\sigma} v(i,j,r).
\]

Since $v$ is win-count oriented, we also have that there is a player value function $p$ such that
\[
V_\sigma=\sum_{i\in N} p(i,w_\sigma(i)),
\]
where $w_{\sigma}(i)$ denotes the number of wins of player $i\in N$ when tournament seeding $\sigma$ is used.

Note that if a player $i$ has $w_\sigma(i)$ wins when seeding $\sigma$ is used, we know by the definition of a knock-out tournament that player $i$ wins one game in each round $r$ with $1\le r\le w_\sigma(i)$.
We now define the following function $v':N\times \mathbb{N}\rightarrow \mathbb{Z}$. For all $i\in N$ and $r\in\mathbb{N}$, we set
\[
v'(i,r)=p(i,r)-p(i,r-1),
\]
where we assume w.l.o.g.\ for all $i\in N$ that $p(i,0)=0$. Using this, we can write
\[
V_\sigma=\sum_{i\in N} p(i,w_\sigma(i))=\sum_{i\in N} \sum_{1\le r\le w_\sigma(i)} v'(i,r).
\]

Instead of summing over the rounds where player $i$ wins, we can sum over the games won by player~$i$. This gives us the following.
\[
\sum_{i\in N} \sum_{1\le r\le w_\sigma(i)} v'(i,r) = \sum_{i\in N} \sum_{(i,j,r),(j,i,r)\in \mathcal{G}_\sigma \text{ s.t.\ }i> j} v'(i,r) .
\]

Now, we can see that for every game, we evaluate $v'$ on the winner and the round of the game. It follows that we can rewrite as follows.
\[
\sum_{i\in N} \sum_{(i,j,r),(j,i,r)\in \mathcal{G}_\sigma \text{ s.t.\ }i> j} v'(i,r) =  \sum_{(i,j,r)\in\mathcal{G}_\sigma} v'(\max(i,j),r).
\]

Summarizing, we have 
\[
V_\sigma=\sum_{(i,j,r)\in\mathcal{G}_\sigma} v(i,j,r) =  \sum_{(i,j,r)\in\mathcal{G}_\sigma} v'(\max(i,j),r).
\]

Note that the above equality also holds for each subtournament. Using this, we prove by induction on $r$ that for all $i,j,r\in N\times N\times \mathbb{N}$ we have $v(i,j,r)=v'(\max(i,j),r)$.

Each subtournament with $r=1$ involves two players, say $i$ and $j$ and w.l.o.g.\ $i>j$. Using the definition of the tournament value (\cref{def:tvalue}) and the assumption that $v$ is win-count oriented, we have that the subtournament value equals $p(i,1)=v(i,j,1)=v(j,i,1)$. 
Hence, we have that $v(i,j,1)=v'(\max(i,j),1)$ for all $i,j\in N$.

Assume $r>1$ and consider a subtournament with $r$ rounds where player $i$ plays against player $j$ in the final round of the subtournament. Let $\sigma$ be the corresponding seeding and let $V^{(r)}_{\sigma}$ denote the value of the subtournament. Then, by definition of the tournament value (\cref{def:tvalue}), we have $V^{(r)}_\sigma=V^{(r-1)}_{\sigma_1}+V^{(r-1)}_{\sigma_2}+v(i, j, r)$, where $V^{(r-1)}_{\sigma_1}$ and $V^{(r-1)}_{\sigma_2}$ denote the values of the subtournaments with $r-1$ rounds that lead up to the final game.
By our arguments above, we also have that $V^{(r)}_\sigma=V^{(r-1)}_{\sigma_1}+V^{(r-1)}_{\sigma_2}+v'(\max(i,j),r)$. It follows that $v(i,j,r)=v'(\max(i,j),r)$ for all $i,j,r\in N$.
\end{proof}

A further natural restriction is the setting where every player has a popularity value, and the value of a game equals the popularity value of the winning player. We call game-value functions capturing this setting \emph{player popularity-based}.

\begin{definition}[Player Popularity-Based Game-Value Function]\label{def:popularity}
        A game-value function $v:N\times N\times \mathbb{N}\rightarrow \mathbb{N}$ is \emph{player popularity-based} if there are \emph{player popularity values} $v_i$ for all players $i\in N$ such that for all $i,j\in N$ and $r\in\mathbb{N}$ we have $v(i,j,r)=v_{\max(i,j)}$.
\end{definition}

We can observe that player popularity-based game-value functions are exactly the win-count oriented game-value functions that are also round oblivious.

\begin{observation}
    A game-value function is player popularity-based if and only if it is win-count oriented and round-oblivious.
\end{observation}

\section{Hardness Results} \label{sec:apx}
In this section, we prove that \probname~is $\mathsf{NP}$-hard and its optimization version is $\mathsf{APX}$-hard even for game-value functions that map to $\{0,1\}$. We also prove that the problem remains $\mathsf{NP}$-hard as well as $\mathsf{APX}$-hard if the game-value functions are round-oblivious and map to three distinct values. 

In order to prove the approximation hardness results in this section, we recall the definition of an $\mathsf{L}$-reduction below.
\begin{definition} [$\mathsf{L}$-reduction, \cite{papadimitriou1988optimization,ausiello2012complexity}] \label{def:lred} A pair of polynomial-time computable functions $(f,g)$ is an $\mathsf{L}$-reduction from an optimization problem $A$ to an optimization problem $B$ if there are
positive constants $\alpha$ and $\beta$ such that for each instance $x$ of $A$, the following hold:
\begin{enumerate}

\item The function $f$ maps instances of $A$ to instances of $B$ such that $\mathsf{opt}_{B}(f(x))\leq \alpha \cdot \mathsf{opt}_{A}(x)$. Here, $\mathsf{opt}_{A}(x)$ denotes the optimal value of instance $x$ of $A$ and $\mathsf{opt}_{B}(f(x))$ denotes the optimal value of instance $f(x)$ of $B$. 
\item The function $g$ maps feasible solutions $y$ of $f(x)$ to feasible solutions $g(y)$ of $x$ such that $\vert \mathsf{opt}_{A}(x)-c_{x}(g(y))\vert \leq \beta \cdot \vert \mathsf{opt}_{B}(f(x))-c_{f(x)}(y)\vert$, where $c_x(\cdot)$ and $c_{f(x)}(\cdot)$ are the cost functions\footnote{These functions compute the value of the feasible solution of the associated problem instance.} of the instances $x$ and $f(x)$, respectively.
\end{enumerate}
\end{definition}

To show the $\mathsf{APX}$-hardness of \probname, we give an $\mathsf{L}$-reduction from \textsc{Max (2,3)-SAT}, which is known to be $\mathsf{NP}$-hard and $\mathsf{APX}$-hard~\cite{ausiello2012complexity}, and is defined below. 

\problemdefopt{\textsc{Max (2,3)-SAT}}{A Boolean formula $\phi$ such that each clause has exactly two literals and each variable appears in at most three clauses.}{Find an assignment to the variables that satisfies the maximum number of clauses of~$\phi$.}

It is worth noting that for a given instance of \textsc{Max (2,3)-SAT}, we can make the assumption that each variable appears in at least 2 clauses. In cases where a variable occurs in at most 1 clause, a preprocessing step, defined below, can be applied to transform the instance into an equivalent one, ensuring that every variable appears in at least 2 clauses.

If a variable does not appear in any clause, it has no impact on the outcome and can be safely eliminated from the instance, and if a variable appears in exactly one clause, then, note that there is an optimal solution such that we can set all variables of this kind to satisfy their only clauses. Based on these observations, we can assume, without loss of generality, that every variable in the given instance of \textsc{Max (2,3)-SAT} appears in at least 2 clauses. Furthermore, since each clause contains exactly 2 variables, this implies that there are at least $n$ clauses, where $n$ is the number of variables. Now, consider the following observation that follows from the above discussion and the fact that any optimal solution of a given instance of \textsc{Max (2,3)-SAT} satisfies at least half of the clauses.
\begin{observation} \label{obs:sat}
If $\mathsf{opt}(\phi)$ denotes the optimal value of an instance $\phi$ of \textsc{Max (2,3)-SAT}, then $\mathsf{opt}(\phi)\geq \frac{n}{2}$, where $n$ is the number of variables in $\phi$.
\end{observation}


\subsection{Hardness for Two Game Values}

We prove that \probname~is $\mathsf{NP}$-hard and its optimization version is $\mathsf{APX}$-hard even for game-value functions that map to $\{0,1\}$. To this end, consider the following construction.

\begin{construction}\label{const:1}
Given an instance $\phi$ of \textsc{Max (2,3)-SAT}, where $\{x_{1},\ldots,x_{n}\}$ is the set of variables and $\{c_{1},\ldots, c_{m}\}$ is the set of clauses (here, note that $m\leq \frac{3n}{2}$), we create an instance $T_{\phi}$ of \probname\ as follows.

Let $x$ be a variable in $\phi$. Then, we create three \emph{variable players} $x, x^T, x^F$ such that $x>x^T>x^F$. We set 
\[
v(x,x^T,1)=v(x,x^F,1)=1. 
\]

Let $c$ be a clause in $\phi$. Then, we create one \emph{clause player} $c$ such that $c$ is weaker than any variable player. Let variable $x$ appear as a literal in $c$. Let this be the $j$th appearance of $x$ (note that $j$ can be 1, 2, or 3) in $\phi$.
If $x$ appears non-negated in $c$, then we set
\[
v(c,x^T,j)=1, 
\]
otherwise, we set 
\[
v(c,x^F,j)=1. 
\]

Let $n'$ be the smallest integer such that $16\cdot n \le 2^{n'}$. Let $p=2^{n'}-16n$. We create $13n+p-m$ \emph{dummy players}, namely, $\{f_{1},\ldots,f_{13n+p-m}\}$, each of which is weaker than any clause player.

The strength ordering of all the players is given below:

$x_{1}>x_{1}^{T}>x_{1}^{F}>x_{2}>\ldots >x_{n}>x_{n}^{T}>x_{n}^{F}>c_{1}>c_{2}>\ldots>c_{m}>f_{1}>\ldots>f_{13n+p-m}$.

For each pair of players $i,j$ and each round $r\in\{1,2,\ldots,n'\}$ where we have not specified the game value above, we set
\[
v(i,j,r)=0.
\]

This finishes the construction, which can easily be computed in polynomial time. Furthermore, we can observe that the constructed game-value function maps to $\{0,1\}$. 

 \begin{figure}[t]
 \centering
    \includegraphics[scale=0.85]{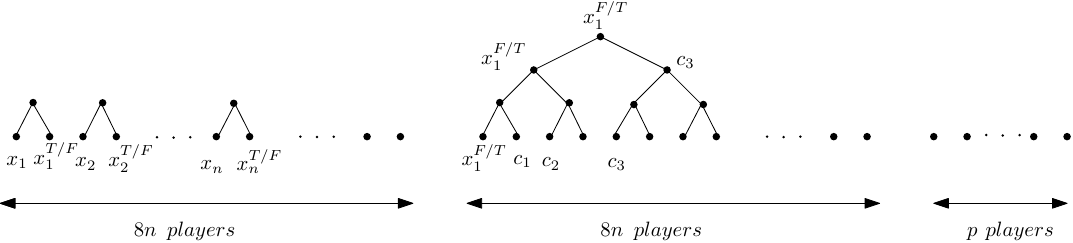}
    \caption{An illustration of the seeding of the players in $T_{\phi}$ in \cref{lem:1}. Here, we assume that $x_{1}^{F/T}$ appears as a literal in clauses $c_{1}$, $c_{2}$, and $c_{3}$.}
    \label{fig:4}
\end{figure}

\end{construction}

Next, we show that \cref{const:1} is an $\mathsf{L}$-reduction.
To this end, we first show that if a formula $\phi$ that is an instance of \textsc{Max (2,3)-SAT} admits an assignment that satisfies $k$ clauses, then there exists a seeding for the instance of \probname\ obtained by applying \cref{const:1} to $\phi$ such that the tournament value is at least $k+n$, where $n$ is the number of variables in $\phi$.

\begin{lemma}\label{lem:1}
Let $\phi$ be an instance of \textsc{Max (2,3)-SAT} with $n$ variables and let $T_\phi$ be the instance of \probname\ obtained by applying \cref{const:1} to $\phi$.
If the formula $\phi$ admits an assignment that satisfies at least $k$ clauses, then $T_{\phi}$ has a seeding corresponding to which the tournament value of $T_{\phi}$ is least $k+n$.
\end{lemma}
\begin{proof}
 Assume that we are given an assignment to the variables $\{x_{1},\ldots,x_{n}\}$ that satisfies at least $k$ clauses of $\phi$.  We construct a seeding for $T_{\phi}$ as follows.

We seed player $x_{i}$ into position $2i-1$. If $x_{i}$ is set to true under the given assignment, then we seed $x_{i}^F$ into position $2i$, and we seed $x_{i}^T$ into position $8n + 8i-7$. Otherwise, we seed $x_{i}^T$ into position $2i$, and we seed $x_{i}^F$ into position $8n + 8i-7$.

Let $c$ be a clause in $\phi$ that is satisfied by one of its literals $\ell$ under the current assignment. Let $x_{i}$ appear in $\ell$, and let this be the $j$th appearance of $x_{i}$ in $\phi$. Then, we seed player $c$ into position $8n + 8i+2^{j-1}-7$.

Finally, we distribute all dummy players and the clause players that are not satisfied under the current assignment arbitrarily among the remaining seed positions. See \cref{fig:4} for an illustration.

We next show that the described seeding results in a tournament value of at least $k+n$. Note that in the first round of the tournament, we have that for every variable $x$ that player $x$ either plays against player $x^F$ (if $x$ is set to true in the satisfying assignment) or player $x^T$ (if $x$ is set to false in the satisfying assignment), since the two players are seeded next to each other. Each of these games has value 1. Hence these games contribute $n$ to the tournament value.

Now, consider clause player $c$ that corresponds to a satisfied clause. Player $c$ is seeded into position $8n + 8i+2^{j-1}-7$. By construction of the seeding, this means that clause $c$ contains the $j$th appearance of $x_{i}$, and that the clause is satisfied by the truth value assigned to variable $x_{i}$ in the satisfying assignment. Furthermore, we have that player $x_{i}^T$ (if $x_{i}$ is set to true in the satisfying assignment) or player $x_{i}^F$ (if $x_{i}$ is set to false in the satisfying assignment) is seeded into position $8n + 8i-7$. This results in player $c$ playing against either player $x_{i}^T$ or player $x_{i}^F$ in round $j$. By construction, the opponent of $c$ is player $x_{i}^T$ if setting $x_{i}$ to true satisfies $c$ and otherwise, player $x_{i}^F$. The corresponding game has value 1. It follows that at least $k$ clause players play a game with value 1. These games contribute at least $k$ to the tournament value. Hence, we get an overall tournament value of at least $n+k$.
\end{proof}

Now, we show the converse direction, that is, if an instance of \probname\ that is obtained by applying \cref{const:1} to some instance $\phi$ of \textsc{Max (2,3)-SAT} admits a seeding that yields a tournament value of $k'$, then $\phi$ admits an assignment that satisfies at least $k'-n$ clauses. 

\begin{lemma}\label{lem:2}
Let $\phi$ be an instance of \textsc{Max (2,3)-SAT} with $n$ variables and let $T_\phi$ be the instance of \probname\ obtained by applying \cref{const:1} to $\phi$.
If $T_{\phi}$ has a seeding corresponding to which the tournament value of $T_{\phi}$ is least $k'$, then the formula $\phi$ admits an assignment that satisfies at least $k'-n$ clauses of $\phi$.
\end{lemma}
\begin{proof}
Assume that the constructed \probname\ instance $T_{\phi}$ has a seeding, say, $\sigma$, that achieves the tournament value $k'$. We give an assignment to $\{x_{1},\ldots,x_{n}\}$ that satisfies at least $k'-n$ clauses of $\phi$. Now, to proceed further, we need the concept of \emph{cheating variable players}. 

For some variable $x$ in $\phi$, we say that the variable players $x^{T}$ and $x^{F}$ \emph{cheat} in a seeding if 
there exist clause players $c$ and $c'$ such that $x^{T}$ plays a game with value 1 against $c$ and $x^{F}$ plays a game with value 1 against $c'$.


In what follows, we will show that if for some variable $x$, the variable players $x^{T}$ and $x^{F}$ cheat in the seeding $\sigma$, then there exists seeding $\sigma'$ for $T_{\phi}$ with tournament value at least $k'$ such that the variable players $x^{T}$ and $x^{F}$ do not cheat. Furthermore, for all variables $x'$ in $\phi$, we have that if the variable players $x'^{T}$ and $x'^{F}$ do not cheat in $\sigma$, then they also do not cheat in $\sigma'$.


Let $x$ be a variable in $\phi$ such that the variable players $x^{T}$ and $x^{F}$ cheat in seeding $\sigma$. Then, there exist clause players $c$ and $c'$ such that $x^{T}$ plays a game with value 1 against $c$ and $x^{F}$ plays a game with value 1 against $c'$.
Now, we have the following claim.
\begin{claim} \label{clm:1}
   The player $x$ must play a 0-value game in the first round.
\end{claim} 

\begin{claimproof}
    Note that the value 1 games that $x$ can play are against $x^{T}$ and $x^{F}$ only. For the sake of contradiction, let us assume w.l.o.g.\ that $x$ plays against $x^{T}$ in the first round, then, since $x$ is stronger than $x^{T}$, it will knock out $x^{T}$, and thus the assumption will not be true (that $x^{T}$ plays a game against clause player $c$).
\end{claimproof} 

So, by \cref{clm:1}, let $x$ play against a player, say, $d \notin \{x^{T},x^{F}\}$ in the first round. 
Furthermore, we can observe that the total number of games with value 1 played by $x^{T}$ and $x^{F}$ is at most 3, since neither of them plays a game with value 1 against player $x$ and there are at most three clause players $c$, $c'$, and potentially $c''$ (if $x$ appears in three clauses) that can play games with value 1 against $x^{T}$ or $x^{F}$. It follows that either $x^{T}$ or $x^{F}$ plays only one game with value 1. Assume w.l.o.g.\ that only $x^{T}$ plays one game with value 1 against $c$.
Now, we will construct a seeding $\sigma'$ from $\sigma$ as follows. 

In $\sigma'$, we seed the players in such a manner that $x$ plays against $x^{T}$ in the first round, and the player playing against $x^{T}$ in the first round in $\sigma$ plays against $d$ in the first round in $\sigma'$. The rest of the seeding remains the same. Note that the tournament value corresponding to $\sigma'$ does not decrease as we lose at most one value 1 game (between $x^{T}$ and clause player $c$) and achieve at least one value 1 game (between $x^{T}$ and $x$). 
Furthermore, note that for all variables $x'$ in $\phi$, we have that if the variable players $x'^{T}$ and $x'^{F}$ are not cheating in $\sigma$, then they are also not cheating in $\sigma'$.
It follows that by repeating the above-described procedure, we can find a seeding for $T_{\phi}$ with tournament value at least $k'$ such that no players cheat.

Finally, we construct an assignment for the variables in $\phi$ as follows. We set the variable $x$ in $\phi$ to true if there is a clause player $c$ that plays a game with value 1 against $x^T$. Otherwise, we set $x$ to false.
The variables that do not get any assignment above are given an arbitrary assignment. 

We claim that the described assignment satisfies at least $k'-n$ clauses.
Note that since we have a seeding without cheating players, we cannot have two clause players $c,c'$ such that $c$ plays a game with value 1 against some variable player $x^T$ and $c'$ plays a game with value 1 against variable player $x^F$. Furthermore, since clause players are weaker than variable players, we have that each clause player can play at most one game with value 1. It follows that in the constructed assignment, we have that a clause is satisfied if and only if the corresponding clause player plays a game of value 1. From the above, we can deduce that the number of satisfied clauses equals the number of games with value 1 that involve a clause player. Notice that we can have at most $n$ games of value 1 that do not involve a clause player: for every variable player $x$, we can have one game in round one against either $x^T$ or $x^F$ that has value one.
Thus, we can say that we have an assignment satisfying at least $k'-n$ clauses of $\phi$.
\end{proof}




\cref{lem:1} and \cref{lem:2} establish the following theorem. 
\begin{theorem} \label{thm:apxhard}
 \probname\ is \textsf{NP}-hard for game-value functions that map to $\{0,1\}$.
\end{theorem}

Furthermore, the following corollary follows immediately from \cref{lem:1} and \cref{lem:2}.

 \begin{corollary} \label{lem:3} If $\mathsf{opt}_{A}(\phi)$ denotes the number of clauses satisfied by an optimal assignment in the given instance $\phi$ of \textsc{Max (2,3)-SAT} and $\mathsf{opt}_{B}(T_{\phi})$ denotes the value of an optimal solution of the constructed \probname\ instance $T_{\phi}$, then $\mathsf{opt}_{B}(T_{\phi})=\mathsf{opt}_{A}(\phi)+n$.
\end{corollary}

Using this corollary, we move our attention towards establishing the $\mathsf{APX}$-hardness. 

\begin{theorem} \label{thm:apxhard1}
The optimization version of \probname\ is \textsf{APX}-hard for game-value functions that map to $\{0,1\}$.
\end{theorem}
\begin{proof}

To obtain the result, we show that \cref{const:1} is an $\mathsf{L}$-reduction. 
Let $\phi$ be an instance of \textsc{Max (2,3)-SAT} with $n$ variables and let $T_\phi$ be the instance of \probname\ obtained by applying \cref{const:1} to $\phi$.


First, note that by \cref{obs:sat}, we have \begin{equation}\label{eq:01} \frac{n}{2}\leq \mathsf{opt}_{A}(\phi).\end{equation} 

By \cref{lem:3}, we have
\begin{equation}\label{eq:02}
\mathsf{opt}_{B}(T_{\phi})=\mathsf{opt}_{A}(\phi)+n
\end{equation}

Combining (\ref{eq:01}) and (\ref{eq:02}), we have 
\begin{equation}\label{eq:03}
     \mathsf{opt}_{B}(T_{\phi})=\mathsf{opt}_{A}(\phi)+n\leq \mathsf{opt}_{A}(\phi)+2\cdot \mathsf{opt}_{A}(\phi)=3\cdot \mathsf{opt}_{A}(\phi).
\end{equation}

From \cref{lem:2}, if the constructed instance of \probname\ has a solution $y$, then the given instance of \textsc{Max (2,3)-SAT} has a solution $g(y)$ such that if $c_A(g(y))$ and $c_B(y)$ denote the cost functions of instances $g(y)$ and $y$, respectively, then 

\begin{equation}\label{eq:004}
c_{A}(g(y))\geq c_{B}(y)-n.
\end{equation}

Now, using (\ref{eq:02}) and (\ref{eq:004}), we have
\begin{equation} \label{eq:05}
\vert (\mathsf{opt}_{A}(\phi)-c_{A}(g(y))\vert=\vert (\mathsf{opt}_{B}(T_\phi)-n)-(c_{B}(y)-n)\vert = \vert \mathsf{opt}_{B}(T_{\phi})-c_{B}(y)\vert.
\end{equation}

From (\ref{eq:03}) and (\ref{eq:05}), it follows that we have an $\mathsf{L}$-reduction with $\alpha=3$ and $\beta=1$ (see \cref{def:lred}). This finishes the proof.
\end{proof}


\subsection{Hardness for Round Obliviousness and Three Game Values} 
We prove that \probname~is $\mathsf{NP}$-hard and its optimization version is $\mathsf{APX}$-hard even for game-value functions that are round oblivious and map to three different values. To this end, consider the following construction.

\begin{construction}\label{const:2}
Given an instance $\phi$ of \textsc{Max (2,3)-SAT}, where $\{x_{1},\ldots,x_{n}\}$ is the set of variables and $\{c_{1},\ldots, c_{m}\}$ is the set of clauses (here, note that $m\leq \frac{3n}{2}$), we create an instance $T_{\phi}$ of \probname\ as follows. 

First, note that we aim to create an instance of \probname\ that has round-oblivious game values. Therefore, we will avoid writing the round number while defining the game-value functions. Furthermore, we will use negative game values in this construction. Towards the end of the section, we argue that these can be removed using \cref{obs:constant}.
For every variable $x$ in $\phi$, we create three \emph{variable players} $x$, $x^T$, $x^F$, and we set 
\begin{equation} \label{eq:1}
    v(x,x^T)=v(x,x^F)=1.
\end{equation}


For every clause $c$ in $\phi$, we create one \emph{clause player} $c$. 
Let variable $x$ appear as a literal in $c$.
If $x$ appears non-negated in $c$, then we set
\[
v(c,x^T)=1,
\]
otherwise, we set 
\[
v(c,x^F)=1.
\]

For every $i\in[n]$, we have 3 \emph{special players} namely, $\widehat{d}_{i}$, $d_{i}$, and $\widetilde{d}_{i}$, and we set
\begin{equation} \label{eq:2}
v(d_{i},\widehat{d}_{i})=v(d_{i},\widetilde{d}_{i})=v(d_{i},x_{i})=0.
\end{equation}

Let $n'$ be the smallest integer such that $16n\le 2^{n'}$. Let $p=2^{n'}-16n$. We create $10n+p-m$ \emph{dummy players}, namely, $\{f_{1},\ldots,f_{10n+p-m}\}$.

The strength ordering of all the players is given below:

$\widehat{d}_{1}>d_{1}>\widetilde{d}_{1}>\widehat{d}_{2}>d_{2}>\widetilde{d}_{2}>\ldots >\widehat{d}_{n}>d_{n}>\widetilde{d}_{n}>x_{1}>x_{1}^{T}>x_{1}^{F}>x_{2}>\ldots >x_{n}>x_{n}^{T}>x_{n}^{F}>c_{1}>c_{2}>\ldots>c_{m}>f_{1}>\ldots>f_{10n+p-m}$.

For every $i\in[n]$ and for every player $Y$ for which we have not already set the value in (\ref{eq:1}) or (\ref{eq:2}), we set 
\[
v(d_{i},Y)=v(x_{i},Y)=-5.
\]

In other words, both $x$ and $d$ can play at most 3 positive value games; otherwise, they have to play a negative value game.

For each pair of players $i,j$, where we have not specified the game value above, we set
\[
v(i,j)=0.
\]

  This finishes the construction which can easily be computed in polynomial time. Furthermore, we can observe that the constructed game-value function maps to $\{0,1,-5\}$ and is round oblivious.

 \begin{figure}[t]
 \centering
    \includegraphics[scale=0.85]{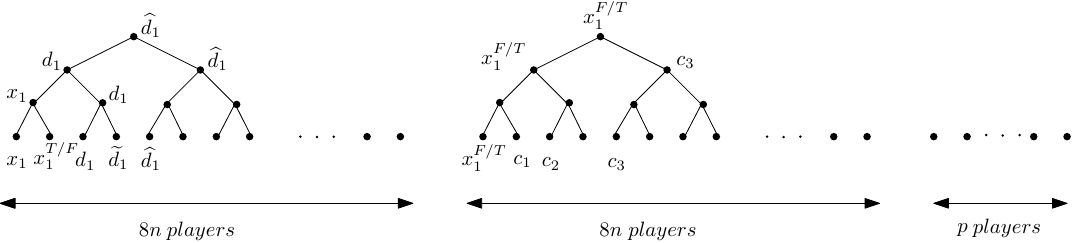}
    \caption{An illustration of the seeding of the players in $T_{\phi}$ in \cref{lem:4}. Here, we assume that $x_{1}^{F/T}$ appears as a literal in clauses $c_{1}$, $c_{2}$, and $c_{3}$.}
    \label{fig:3}
\end{figure}
    
\end{construction}

Next, we show that \cref{const:2} is an $\mathsf{L}$-reduction.
To this end, we first show that if a formula $\phi$ that is an instance of \textsc{Max (2,3)-SAT} admits an assignment that satisfies $k$ clauses, then there exists a seeding for the instance of \probname\ obtained by applying \cref{const:2} to $\phi$ such that the tournament value is at least $k+n$, where $n$ is the number of variables in $\phi$.

\begin{lemma}\label{lem:4}
Let $\phi$ be an instance of \textsc{Max (2,3)-SAT} with $n$ variables and let $T_\phi$ be the instance of \probname\ obtained by applying \cref{const:2} to $\phi$.
If the formula $\phi$ has an assignment that satisfies at least $k$ clauses, then $T_{\phi}$ has a seeding corresponding to which the tournament value of $T_{\phi}$ is least $k+n$.
\end{lemma}
\begin{proof}
Assume that we are given an assignment to the variables $\{x_{1},\ldots,x_{n}\}$ that satisfies at least $k$ clauses of $\phi$. We construct a seeding for $T_{\phi}$ as follows.

We seed player $x_{i}$ into position $8i-7$. If $x_{i}$ is set to true in the satisfying assignment, then we seed $x_{i}^F$ into position $8i-6$, and we seed $x_{i}^T$ into position $8n + 8i-7$. Otherwise, we seed $x_{i}^T$ into position $8i-6$, and we seed $x_{i}^F$ into position $8n + 8i-7$. Next, we seed $d_{i}$ into position $8i-5$, $\widetilde{d}_{i}$ into position $8i-4$, and $\widehat{d_{i}}$ into position $8i-3$.

Let $c$ be a clause in $\phi$ that is satisfied by one of its literals $\ell$ under the current assignment. Let $x_{i}$ appear in $\ell$, and let this be the $j$th appearance of $x_{i}$ in $\phi$ (note that $j$ can be 1, 2, or 3). Then, we seed player $c$ into position $8n + 8i+2^{j-1}-7$. Let $c'$ be a clause in $\phi$ that is not satisfied by the current assignment. Then, we distribute $c'$ and all dummy players arbitrarily among the remaining seed positions. See \cref{fig:3} for an illustration.

We next show that the described seeding results in a tournament value of at least $k+n$. Note that in the first round of the tournament, we have that for every variable $x$ that player $x$ either plays against player $x^F$ (if $x$ is set to true in the satisfying assignment) or player $x^T$ (if $x$ is set to false in the satisfying assignment), since the two players are seeded next to each other. Each of these games has value 1. Hence these games contribute $n$ to the tournament value.

Now, consider clause player $c$ that corresponds to a satisfied clause. Player $c$ is seeded into position $8n + 8i+2^{j-1}-7$. By construction of the seeding, this means that clause $c$ contains the $j$th appearance of the variable $x_{i}$ of $\phi$, and that the clause is satisfied by the truth value assigned to variable $x_{i}$ in the satisfying assignment. Furthermore, we have that player $x_{i}^T$ (if $x_{i}$ is set to true in the satisfying assignment) or player $x_{i}^F$ (if $x_{i}$ is set to false in the satisfying assignment) is seeded into position $8n + 8i-7$. This results in player $c$ playing against either player $x_{i}^T$ or player $x_{i}^F$ in round $j$. By construction, the opponent of $c$ is player $x_{i}^T$ if setting $x_{i}$ to true satisfies $c$ and otherwise, player $x_{i}^F$. The corresponding game has value 1. It follows that every satisfied clause player plays one game with value 1. These games contribute at least $k$ to the tournament value. Hence, we get an overall tournament value of at least $k+n$. Note that no game carrying a negative value occurs.
\end{proof}

Now, we show the converse direction, that is, if an instance of \probname\ that is obtained by applying \cref{const:2} to some instance $\phi$ of \textsc{Max (2,3)-SAT} admits a seeding that yields a tournament value of $k'$, then $\phi$ admits an assignment that satisfies at least $k'-n$ clauses. 

\begin{lemma}\label{lem:5}
Let $\phi$ be an instance of \textsc{Max (2,3)-SAT} with $n$ variables and let $T_\phi$ be the instance of \probname\ obtained by applying \cref{const:2} to $\phi$.
If $T_{\phi}$ has a seeding corresponding to which the tournament value of $T_{\phi}$ is least $k+n$, then the formula $\phi$ admits an assignment that satisfies at least $k$ clauses of $\phi$.
\end{lemma}

\begin{proof}
Assume that the constructed \probname\ instance $T_{\phi}$ has a seeding, say, $\sigma$, that achieves the tournament value $k'$.  We construct an assignment for $\{x_{1},\ldots,x_{n}\}$ that satisfies at least $k'-n$ clauses of $\phi$. 
Now, to proceed further, we need the concept of \emph{cheating variable players}, which is similar to the one we use in the proof of \cref{thm:apxhard}.

For some variable $x$ in $\phi$, we say that the variable players $x^{T}$ and $x^{F}$ \emph{cheat} in a seeding if:
\begin{enumerate}
    \item there exist clause players $c$ and $c'$ such that $x^{T}$ plays a game with value 1 against $c$ and $x^{F}$ plays a game with value 1 against $c'$, or
    \item both $x^{T}$ and $x^{F}$ play a game against player $x$.
\end{enumerate}

In what follows, we analyze seeding with cheating players and establish some properties of those seedings.

 We first show that if players $x_{i}^{T}$ and $x_{i}^{F}$ cheat in the second sense, there must be at least one game with value $-5$ that involves either $x_{i}$ or $d_{i}$. 
Intuitively, this will allow us to set variable $x$ to some arbitrary value.

\begin{claim}\label{claim:negativegames1}
    Let $x_i$ be a variable in $\phi$ and assume that both players $x_i^{T}$ and $x_i^{F}$ play a game against player $x_i$ in the seeding $\sigma$. Then, player $x_i$ or player $d_i$ plays a game with value $-5$.
\end{claim}
\begin{claimproof}
 If $x_{i}$ plays against some player $\hat{d}_{j}$, $d_{j}$ or $\tilde{d}_{j}$ in the first round, then 
 $x_{i}$ will get knocked out by that player in the first round and there are no games of $x_i$ against $x_i^{T}$ and $x_i^{F}$. 
 

Next, let $x_{i}$ play w.l.o.g.\ against $x_{i}^{T}$ in the first round. Then, $x_{i}^{T}$ gets knocked out in the first round. Now, there are two possibilities. If $x_{i}$ plays against some player $\hat{d}_{j}$, $d_{j}$ or $\tilde{d}_{j}$ in the second round, then again, as $x_{i}$ gets knocked out by that player.
In this case, no further game between $x_{i}^{F}$ and $x_i$ can happen. 

Now, w.l.o.g.\ let $x_{i}$ play against $x_{i}^{T}$ in the first round and play against $x_{i}^{F}$ in some later round, say $r\ge 2$ (here, we note that $r$ cannot be the last round as $\widehat{d}_{1}$ must play in the last round). Then, there are at least two cheating games. Now, we will show that there exists a negative value game, that is, a game with value $-5$. If $x_{i}$ does not play against $d_{i}$ in round $r+1$, then the game played by $x_{i}$ in round $r+1$ must have value $-5$, and we are done. So, assume that $x_{i}$ plays against $d_{i}$ in round $r+1$. This implies that $d_{i}$ must have won at least two games (in the first and second rounds). But, by the definition of game-value functions, we know that at least one of these games must be a negative value game, and we are done. 
Lastly, if $x_{i}$ plays against any of the remaining players in the first round, then it must carry a negative value, and we are done again.  

We can conclude that if both $x_{i}^{T}$ and $x_{i}^{F}$ play a game against $x_i$, then $x_i$ or $d_i$ play a game with value $-5$. 
\end{claimproof}

Next, we analyze the case where a variable player $x_i^T$ (the case for $x_i^F$ is analogous) plays against player $x_i$ and also plays against some clause player $c$. Note that it is not considered cheating by the provided definition unless player $x_i^F$ also plays against $x_i$ or some variable player $c'$. Nevertheless, we show that in this case also, there must be a game with a negative value. This allows us later to concentrate on the case where $x_{i}^{T}$ and $x_{i}^{F}$ cheat, but neither player plays against $x_i$.

\begin{claim}\label{claim:negativegames2}
    Let $x_i$ be a variable in $\phi$ and assume that player $x_i^{T}$ plays a game against player $x_i$ and a game against a clause player $c$ in the seeding $\sigma$. Then, player $x_i$ or player $d_i$ plays a game with value $-5$.
\end{claim}

\begin{claimproof}
    First, note that if player $x_i^F$ also plays against $x_i$, then by \cref{claim:negativegames1} we have that player $x_i$ or player $d_i$ plays a game with value $-5$. Hence, assume that player $x_i^F$ does not play against player $x_i$.

    Note that if player $x_i^T$ plays both against $x_i$ and some clause player $c$, then the game between $x_i^T$ and $x_i$ cannot occur in round one, since $x_i$ is stronger than $x_i^T$. It follows that in round one, player $x_i$ plays against a player weaker than itself that is different from $x_i^F$ (since we assume $x_i$ does not play both against $x_i^T$ and $x_i^F$). By the definition of the game-value function, this game has value $-5$.
\end{claimproof}

Finally, we make the following observation on cheating players $x^T$ and $x^F$ that both do not play against $x$.

\begin{claim}\label{claim:cheat1}
    Let players $x^T$ and $x^F$ be cheating in $\sigma$ but neither $x^T$ nor $x^F$ play against player $x$. 
    Then, at least one of the players $x^T$ and $x^F$ plays at most one game with value 1.
\end{claim}
\begin{claimproof}
Let players $x^T$ and $x^F$ be cheating in $\sigma$ but neither $x^T$ nor $x^F$ play against player $x$. Then, by the definition of cheating players, there exist clause players $c$ and $c'$ such that $x^{T}$ plays a game with value 1 against $c$ and $x^{F}$ plays a game with value 1 against $c'$.
Furthermore, we can observe that the total number of games with value 1 played by $x^{T}$ and $x^{F}$ is at most 3, since neither of them plays a game with value 1 against player $x$ and there are at most three clause players $c$, $c'$, and potentially $c''$ (if $x$ appears in three clauses) that can play games with value 1 against $x^{T}$ or $x^{F}$. It follows that either $x^{T}$ or $x^{F}$ plays only one game with value 1. 
%
\end{claimproof}

We now construct an assignment for $\phi$ as follows. Let $x$ be a variable in $\phi$.
\begin{itemize}
    \item If players $x^T$ and $x^F$ are not cheating and player $x^T$ plays against some clause player $c$, we set $x$ to true. Otherwise, we set $x$ to false.
    \item If players $x^T$ and $x^F$ cheat and neither of them plays against $x$, then by \cref{claim:cheat1} one of $x^T$ and $x^F$ plays at most one value 1 game. If that player is $x^T$, we set $x$ to false. Otherwise, we set $x$ to true. 
    \item If none of the above is true, we set $x$ (arbitrarily) to true.
\end{itemize}

We claim that the above-described assignment satisfies at least $k'-n$ clauses.
To this end, we subtract from $k'$ every point of the tournament value that we cannot attribute to a satisfied clause. Note that every game with positive value involves either player $x^T$ or player $x^F$ for some variable $x$. Furthermore, note that every clause player can play at most one game of value 1. 

Consider players $x_i^T$ and $x_i^F$ corresponding to some variable $x_i$, then we have the following cases.
\begin{itemize}
    \item If $x_i^T$ and $x_i^F$ are not cheating, at most one of them plays a game of value 1 against player $x_i$, say w.l.o.g.\ $x_i^T$. All other games with value 1 are between players $x_i^F$ and some clause players. By construction of the instance and since in this case we set $x_i$ to false, all clauses corresponding to clause players that play a game with value 1 against $x_i^F$ are satisfied.
    Hence, there is at most one point of value that we cannot attribute to a satisfied clause.
    \item If players $x_i^T$ and $x_i^F$ cheat and neither of them plays against $x_i$, then we know by \cref{claim:cheat1} that one of $x_i^T$ and $x_i^F$ plays at most one value 1 game, say, w.l.o.g.\ player $x_i^T$. All other games with value 1 are between players $x_i^F$ and some clause players. By construction of the instance and since in this case we set $x_i$ to false, all clauses corresponding to clause players that play a game with value 1 against $x_i^F$ are satisfied.
    Hence, again, there is at most one point of value that we cannot attribute to a satisfied clause.
    \item If none of the above is true, then by \cref{claim:negativegames1} and \cref{claim:negativegames2}, we know that player $x_i$ or player $d_i$ plays a game with value $-5$. Note that there are at most three clause players $c$, $c'$, and potentially $c''$ (if $x$ appears in three clauses) that can play games with value 1 against $x_i^{T}$ or $x_i^{F}$. Furthermore, there are at most two additional games of $x_i^{T}$ and $x_i^{F}$ against $x$. It follows that players $x_i^{T}$ and $x_i^{F}$ can have at most 5 games with value 1. However, since we also have a game of value $-5$ that we can attribute to variable $x_i$, since it is played by player $x_i$ or player $d_i$, the contribution of all value 1 games of $x_i^{T}$ and $x_i^{F}$ to the overall tournament value is canceled out by the one game with value $-5$. Hence, in this case, there are no points of value from the overall tournament value that we need to attribute to some satisfied clauses.
\end{itemize}

Since there are $n$ variables and for each variable, there is at most one point of value that we cannot attribute to a satisfied clause, it follows that at least $k'-n$ clauses of $\phi$ are satisfied by the described assignment.
\end{proof}

\cref{lem:4} and \cref{lem:5} establish the following theorem. 

\begin{theorem} \label{thm:apx2}
\probname\ is \textsf{NP}-hard for round-oblivious game-value functions that map to three distinct values.
\end{theorem}





Furthermore, the following corollary follows immediately from \cref{lem:1} and \cref{lem:2}.

 \begin{corollary} \label{lem:6} If $\mathsf{opt}_{A}(\phi)$ denotes the number of clauses satisfied by an optimal assignment in the given instance $\phi$ of \textsc{Max (2,3)-SAT} and $\mathsf{opt}_{B}(T_{\phi})$ denotes the value of an optimal solution of the constructed \probname\ instance $T_{\phi}$, then $\mathsf{opt}_{B}(T_{\phi})=\mathsf{opt}_{A}(\phi)+n$.
\end{corollary}

Using this corollary, we move our attention towards establishing the $\mathsf{APX}$-hardness. 

\begin{theorem} \label{thm:apx3}
The optimization version of \probname\ is \textsf{APX}-hard for round-oblivious game-value functions that map to three distinct values.
\end{theorem}

\begin{proof}
To obtain the result, we show that \cref{const:2} is an $\mathsf{L}$-reduction. 
Let $\phi$ be an instance of \textsc{Max (2,3)-SAT} with $n$ variables and let $T_\phi$ be the instance of \probname\ obtained by applying \cref{const:2} to $\phi$.

First, note that by \cref{obs:sat}, we have \begin{equation}\label{eq:201} \frac{n}{2}\leq \mathsf{opt}_{A}(\phi).\end{equation} 

By \cref{lem:6}, we have
\begin{equation}\label{eq:202}
\mathsf{opt}_{B}(T_{\phi})=\mathsf{opt}_{A}(\phi)+n
\end{equation}

Combining (\ref{eq:201}) and (\ref{eq:202}), we have 
\begin{equation}\label{eq:203}
     \mathsf{opt}_{B}(T_{\phi})=\mathsf{opt}_{A}(\phi)+n\leq \mathsf{opt}_{A}(\phi)+2\cdot \mathsf{opt}_{A}(\phi)=3\cdot \mathsf{opt}_{A}(\phi).
\end{equation}

From \cref{lem:5}, if the constructed instance of \probname\ has a solution $y$, then the given instance of \textsc{Max (2,3)-SAT} has a solution $g(y)$ such that if $c_A(g(y))$ and $c_B(y)$ denote the cost functions of instances $g(y)$ and $y$, respectively, then 

\begin{equation}\label{eq:104}
c_{A}(g(y))\geq c_{B}(y)-n.
\end{equation}


Now, using (\ref{eq:202}) and (\ref{eq:104}), we have
\begin{equation} \label{eq:205}
\vert (\mathsf{opt}_{A}(\phi)-c_{A}(g(y))\vert=\vert (\mathsf{opt}_{B}(T_\phi)-n)-(c_{B}(y)-n)\vert = \vert \mathsf{opt}_{B}(T_{\phi})-c_{B}(y)\vert.
\end{equation}

From (\ref{eq:203}) and (\ref{eq:205}), it follows that we have an $\mathsf{L}$-reduction with $\alpha=3$ and $\beta=1$ (see \cref{def:lred}). This finishes the proof.
\end{proof}


Here, it is worth mentioning that we have used a negative game value in \cref{const:2}, mainly to simplify proofs. Nevertheless, it is crucial to acknowledge that \cref{thm:apx2} and \cref{thm:apx3} remain valid even when all the game values are non-negative. For this purpose, consider the following.
\begin{remark}
If we replace the game-value function $v$ with $v'(i,j)=v(i,j)+6$ in \cref{const:2}, then all the game values are non-negative and \cref{thm:apx2} and \cref{thm:apx3} still hold.
\end{remark}
\begin{proof}
Let $\phi$ be an instance of \textsc{Max (2,3)-SAT} with $n$ variables and let $T_\phi$ be the instance of \probname\ obtained by applying \cref{const:2} to $\phi$.
Let $T'_{\phi}$ be the instance of \probname~obtained from $T_{\phi}$ by replacing the game-value function $v$ with $v'(i,j)=v(i,j)+6$. Note that exactly $2^{n'}-1$ games occur in $T'_{\phi}$ with $2^{n'}\in\Theta(n)$. Now, by the similar arguments, as we used in the proofs of \cref{lem:4} and \cref{lem:5} (that lead to \cref{lem:6}), we have the following. 

The formula $\phi$ has an assignment that satisfies at least $k$ clauses if and only if $T'_{\phi}$ has a seeding corresponding to which the tournament value of $T'_{\phi}$ is least $(2^{n'}-1)\cdot 6+k+n$. Note that since $2^{n'}\in\Theta(n)$ we have that $(2^{n'}-1)\cdot 6+k+n=k+c\cdot n$ for some constant $c$.




Thus, we also have the following. 
If $\mathsf{opt}_{A}(\phi)$ denotes the number of clauses satisfied by an optimal assignment in the given instance $\phi$ of \textsc{Max (2,3)-SAT} and $\mathsf{opt}_{B}(T'_{\phi})$ denotes the value of an optimal solution of the constructed \probname\ instance $T'_{\phi}$, then $\mathsf{opt}_{B}(T'_{\phi})=\mathsf{opt}_{A}(\phi)+c\cdot n$.

Now, by analogous arguments as in the proof of \cref{thm:apx3}, we get that a modified version of \cref{const:2} that uses game-value function $v'(i,j)=v(i,j)+6$ instead of $v(i,j)$ is an  $\mathsf{L}$-reduction with $\alpha=2c+1$ and $\beta=1$ (see \cref{def:lred}). 
\end{proof}

\section{Algorithmic Results}\label{sec:algos}

In this section, we present several algorithmic results. We start in \cref{apx:algo} by describing a polynomial time $(1/\log n)$-approximation algorithm for \probname\ with round-oblivious game-value functions. We proceed by giving a quasi-polynomial time algorithm for \probname\ with win-count oriented game-value functions in \cref{sec:alg} and a linear time greedy algorithm as well as an FPT-algorithm for further restricted settings in \cref{sec:greedy}. Finally, we give a further FPT-algorithm in \cref{sec:fptvc} for more general instances of \probname\ parameterized by the so-called ``size of the set of influential players''.

\subsection{A Polynomial-Time $(1/\log n)$-Approximation Algorithm}\label{apx:algo}

In this section, we show the existence of a $(1/\log n)$-approximation algorithm for the optimization version of \probname\ with a round-oblivious game-value function that runs in polynomial time.
Formally, we prove the following theorem.

\begin{theorem} \label{thm:approx}

The optimization version of \probname\ with a round-oblivious game-value function admits a polynomial-time $(1/\log n)$-approximation algorithm.
\end{theorem}
\begin{proof}
Consider Algorithm $\mathcal{A}$ that, given an instance T of the optimization version of \probname\ with the player set $N$ and with a round-oblivious game-value function $v$, performs the following steps (1-5):
\begin{enumerate}
\item First, construct a complete undirected graph $G$ that has $N$ as its vertex set and an edge-weight function $w:E(G)\rightarrow \mathbb{N}$ defined as follows: 
\[
w(\{i,j\})=\max({v(i,j),v(j,i)}).
\]
\item Compute a maximum weight matching $M_{\max}\subseteq E(G)$ in $G$. 
\item Let $e_\ell=\{i,j\}\in M_{\max}$ denote the $\ell$th edge in $M_{\max}$ and let w.l.o.g.\ $v(i,j)\ge v(j,i)$. Then, define a seeding $\sigma$ such that player $i$ is seeded into position $2\ell-1$ and player $j$ is seeded into position $2\ell$.
\item Let $\overline{V}_{\max}$ denote the set of vertices that are not saturated by $M_{\max}$. For every $i\in \overline{V}_{\max}$, seed player $i$ arbitrarily into any of the remaining seed positions of $\sigma$(which are, $\{2|M_{\max}|+1,\ldots,n\}$).
\item Return $\sigma$.
\end{enumerate}

Now, we claim that if $W_{\max}$ denotes the tournament value of T corresponding to the seeding $\sigma$ (returned by Algorithm $\mathcal{A}$), then $W_{\max}$ is an $(1/\log n)$-approximation of the maximum achievable tournament value of T. Before we prove the approximation bound, note that Algorithm $\mathcal{A}$ runs in polynomial time, as maximum weight matchings can be computed in polynomial time~\cite{galil1986efficient}. 
 
 In the remainder, we show that Algorithm $\mathcal{A}$ is a $(1/\log n)$-approximation algorithm for \probname. 
Let $V^\star$ denote the highest achievable tournament value of T.
Observe that we have 
\[
V^\star\ge W_{\max}.
\]

Now, let $V^\star_r$ denote the value of games played in round $r\in[\log n]$ in the seeding that achieves tournament value $V^\star$. We have that 
\[
V^\star_r\le W_{\max},
\]
since otherwise the games played in round $r$ would correspond to edges in $G$ that form a matching with a weight larger than $W_{\max}$ (because the value of $W_{\max}$ is exactly equal to the weight of $M_{\max}$).

If follows that 
\[
V^\star=\sum_{r=1}^{\log n}V^\star_r \le \log n\cdot W_{\max}.
\]

Overall, we have that 
\[
\frac{1}{\log n}V^\star\le W_{\max}\le V^\star,
\]
and hence the theorem follows.
\end{proof}

%

We remark that the bound presented in \cref{thm:approx} is indeed tight for the presented algorithm. To illustrate this, consider the example shown in \cref{fig:03}, which demonstrates a scenario where the algorithm's performance matches the bound.
 \begin{figure}[t]
 \centering
    \includegraphics[scale=0.35]{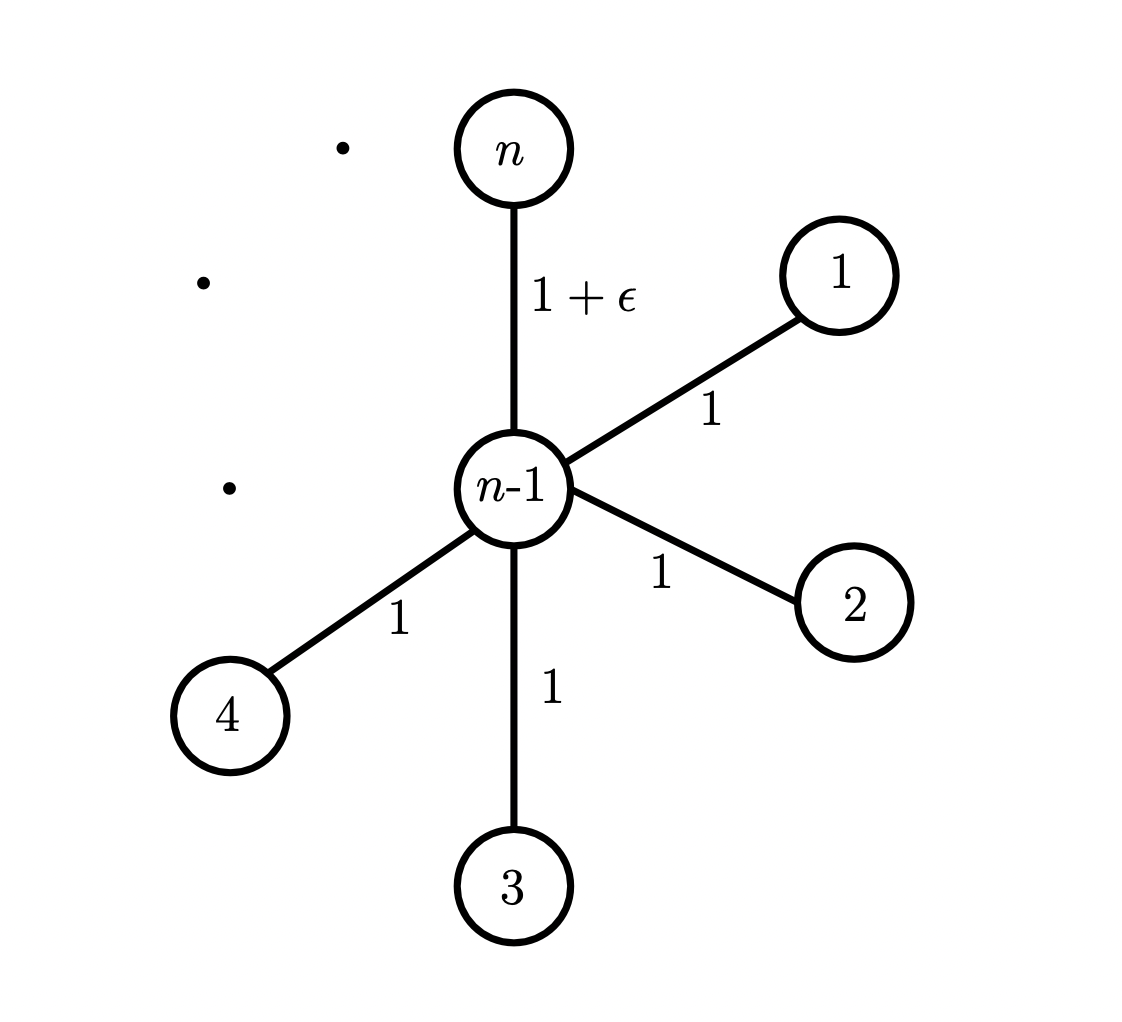}
    \caption{Here, we are given an instance of the optimization version of \probname~with $n$ players and round-oblivious game-value function defined as: $v(n-1,j)=1$ for all $j\in[n-2]$, $v(n-1,n)=1+\varepsilon$ (where $\varepsilon>0$ is a constant), and $v(i,j)=0$ for all $i,j\in [n]\setminus \{n-1\}$. Now, note that Algorithm $\mathcal{A}$ (given in \cref{thm:approx}) will return a seeding with tournament value $1+\epsilon$ (by making sure that players $n$ and $n-1$ play against each other in the first round). However, the optimal tournament value for this given instance is $\log n+\epsilon$ (one of the possible seedings is $(n-1,1,2,\ldots,n-2,n)$).}
    \label{fig:03}
\end{figure}

\subsection{A Quasi-Polynomial-Time Algorithm for Win-Count Oriented Game-Value Functions}\label{sec:alg}



For \probname\ with a win-count oriented game-value function (see \cref{def:wincount}), we present a quasipolynomial-time algorithm.

\begin{theorem}\label{thm:quasipoly}
\probname\ with a win-count oriented game-value function can be solved in $n^{\OO(\log n)}$ time.
\end{theorem}

To describe the algorithm for \cref{thm:quasipoly}, we introduce some additional terminology and concepts. Intuitively, we want to iterate through all players from the strongest to the weakest. We know that the strongest player wins the tournament, so we can assume w.l.o.g.\ that it is seeded into some fixed position, say one.


Once the winner of a tournament with $n$ seed positions is placed in the seeding, we say that $\log n$ subtournaments \emph{open up}, one for each round $r\in[\log n-1]$ of the tournament except the last one and one degenerate subtournament with only one player (if $n>1$, for ``round zero''). The subtournament for round $r$ has the $2^{r}$ seed positions $\{2^{r}+1,2^{r}+2,\ldots, 2^{r+1}\}$. The degenerate subtournament for ``round zero'' has seed position~2. The winner of each of the subtournaments is seeded into the smallest seed position of that subtournament and loses against the winner of the overall tournament in the respective round. See \cref{fig:5} for an illustration.

 \begin{figure}[t]
 \centering
    \includegraphics[scale=1]{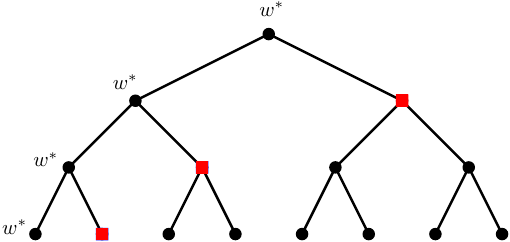}
    \caption{A tournament with $n=8$ players where the winning player $w^*$ is seeded into position one. The red vertices represent the roots of the three subtournaments that open.}
    \label{fig:5}
\end{figure}

For each subtournament, we now have an analogous situation. Since we seed players iteratively and decreasingly to their strength, we know that the first player seeded into a subtournament wins the subtournament. Once a winning player for a subtournament for some round $r$ is seeded, we say that the subtournament is \emph{closed}, and $r$ new subtournaments open up, one for each of the $r-1$ rounds of the subtournament for round $r$ and one for ``round zero''.

Now, after a prefix of the players is seeded (according to their strength ordering), a certain set of subtournaments is open. We classify subtournaments by their size or, equivalently, by their number of rounds. We call a vector $\mathcal{X}=(x_0, x_1, \ldots, x_{\log n-1})\in \mathbb{N}^{\log n}$ a \emph{subtournament profile}, where $x_r$ quantifies the number of open subtournaments with $r$ rounds. 

We define the following dynamic program $T:\mathbb{N}^{\log n+1}\rightarrow\mathbb{Z}$, where, intuitively, $T[\ell,\mathcal{X}]$ quantifies the maximum value achievable by seeding the $\ell$ strongest players while obtaining subtournament profile $\mathcal{X}$.
Let $N=\{1, 2, \ldots, 2^{n'}\}$ be the set of players. We initialize $T$ as follows.
\[
T[1,x_0,x_1,\ldots,x_{n'-1}]=
\begin{cases}
      p(2^{n'},n') & \text{if}\ x_0=x_1=\ldots=x_{n'-1}=1, \\
      -\infty & \text{otherwise.}
    \end{cases}
\]

The remaining entries of $T$ are recursively defined as follows.
\[
T[\ell,x_0,x_1,\ldots,x_{n'-1}]=\max_{r\in \{0,1,\ldots,n'-1\}} (T[\ell-1,x_0-1,\ldots,x_{r-1}-1,x_r+1,x_{r+1},\ldots, x_{n'-1}]+p(2^{n'}-\ell+1,r)).
\]

We first prove the correctness of the dynamic program.

\begin{lemma}\label{lem:dpcorrectness}
$T[\ell,\mathcal{X}]$ is the maximum value achievable by seeding the $\ell$ strongest players while obtaining subtournament profile $\mathcal{X}$.
\end{lemma}
\begin{proof}
We prove the statement by induction on $\ell$. Initially, for $\ell=1$, we have
\[
T[1,x_0,x_1,\ldots,x_{n'-1}]=
\begin{cases}
      p(2^{n'},n') & \text{if}\ x_0=x_1=\ldots=x_{n'-1}=1, \\
      -\infty & \text{otherwise.}
    \end{cases}
\]
Note that we start with the strongest player $i=2^{n'}$. Hence, we have that $i$ wins all games. It follows that for each round $r<n'$, one subtournament with $r$ rounds is opened. Furthermore, one degenerate tournament for round zero is opened. Hence, the only achievable subtournament profile is $\mathcal{X}=(1,1,\ldots, 1)$ and the value achieved in this case is $p(2^{n'},n')$.

Now, assume that $\ell > 1$ and let $\mathcal{X}=(x_0,x_1,\ldots,x_{n'-1})$ be a subtournament profile.
Note that when seeding a player as a winner of a subtournament with $r$ rounds, one subtournament with $r$ rounds is closed, and $r$ subtournaments with $0,1,\ldots,r-1$ rounds, respectively, are opened. It follows that to achieve the subtournament profile $\mathcal{X}=(x_0,x_1,\ldots, x_{n'-1})$ by seeding the $\ell$th player as a winner of a subtournament with $r$ rounds, the previous subtournament profile must be $\mathcal{X}'=(x_0-1,x_1-1,\ldots,x_{r-1}-1,x_r+1,x_{r+1},x_{r+2},\ldots, x_{n'-1})$. Seeding the $\ell$th strongest player as a winner of a subtournament with $r$ rounds increases the tournament value by $p(2^{n'}-\ell+1,r)$. By induction, we have that the optimal tournament value achievable by seeding player $\ell$ as a winner of a subtournament with $r$ rounds and obtaining subtournament profile $\mathcal{X}$ is $T[\ell-1,x_0-1,\ldots,x_{r-1}-1,x_r+1,x_{r+1},\ldots, x_{n'-1}]+p(2^{n'}-\ell+1,r)$. It follows that the optimal tournament value achievable by seeding the $\ell$ strongest players while obtaining subtournament profile $\mathcal{X}$ is the maximum over all the possibilities of how to seed the $\ell$th strongest player, which is
\[
\max_{r\in \{0,1,\ldots,n'-1\}} (T[\ell-1,x_0-1,\ldots,x_{r-1}-1,x_r+1,x_{r+1},\ldots, x_{n'-1}]+p(2^{n'}-\ell+1,r)).
\]
It follows that the dynamic program is correct.
\end{proof}
To prove \cref{thm:quasipoly}, it remains to bound the running time of the dynamic program.
\begin{proof}[Proof of \cref{thm:quasipoly}]
    Consider the dynamic programming table $T:\mathbb{N}^{\log n+1}\rightarrow\mathbb{N}$ described above. By \cref{lem:dpcorrectness} we have that we can find the optimal solution by computing $T[2^{n'},0,0,\ldots,0]$; when all players are seeded, all subtournaments are closed. Each entry of the table can be computed in time $\OO(\log n)$ time, since we compute the maximum over $\OO(\log n)$ values, each of which can be looked up in constant time. From the definition of the table, it follows that its size is in $n^{\OO(\log n)}$, since we consider $\OO(\log n)$ different sizes of subtournaments and $\OO(n)$ subtournaments of each size can be open. \cref{thm:quasipoly} follows.
\end{proof}

\subsection{Algorithms for Player Popularity-Based Game-Value Functions}\label{sec:greedy}

Here, we present a linear time greedy algorithm for special cases of player popularity-based game-value functions where there are only two different player popularity values. Furthermore, we give an FPT algorithm for the ``disagreement'' between the player popularity values and the strength ordering, which, intuitively, measures how different the strength ordering is from the ordering obtained by sorting the players (descending) by their player popularity value. We give a formal definition when we introduce the algorithm.

The algorithms presented here use the concept of open and closed subtournaments introduced in the previous section. We start with presenting a greedy algorithm for the case where there are only two different player popularity values.

\begin{theorem}\label{thm:twovalues}
    \probname\ with a player popularity-based game-value function is solvable in linear time if the game-value function is player popularity-based and there are only two different player popularity values.
\end{theorem}

\begin{proof}
Let $a$ and $b$ denote the two distinct player popularity values with $a>b$. We call players with popularity value $a$ \emph{popular} and all other players \emph{unpopular}.

We provide a simple greedy algorithm for this case that uses the concept of open and closed subtournaments. 
During the execution of the algorithm, we keep track of how many subtournaments of which size are currently open.
Note that since we are given a tournament with $n$ players, initially, one subtournament (with $\log n$ rounds) is open. Now, our algorithm performs the following steps.   
\begin{enumerate}
\item The initial tournament value is set to 0, and one subtournament with $\log n$ rounds is open. 
\item Starting with the strongest player, iterate through the players in order of their strength.
\item Let player $i$ be the player considered in the current iteration. Let $r$ denote the largest number of rounds of an open subtournament, and let $r'$ denote the smallest number of rounds of an open subtournament.
\begin{enumerate}
\item[3.1] If player $i$ is popular:
\begin{enumerate}
\item[3.1.1] Close one open subtournament with $r$ rounds and open $r$ new subtournaments as follows: for each round $r''\in[r-1]$, we open one subtournament with $r''$ rounds. Furthermore, if $r\neq 0$, we open one degenerate subtournament for round zero.
\item[3.1.2] Add $r\cdot a$ to the current tournament value.
\end{enumerate} 
\item[3.2] If player $i$ is unpopular:
\begin{enumerate}
\item[3.2.1] Close one open subtournament with $r'$ rounds and open $r'$ new subtournaments as follows: for each round $r''\in[r'-1]$, we open one subtournament with $r''$ rounds. Furthermore, if $r'\neq 0$, we open one degenerate subtournament for round zero.
\item[3.2.2] Add $r'\cdot b$ to the current tournament value.
\end{enumerate} 
\end{enumerate}
\end{enumerate}

First, we argue that the described algorithm runs in linear time. We iterate through the players and consider each player once. Furthermore, since there are total $n$ subtournaments, and each subtournament is opened and closed exactly once in our algorithm, we can conclude that the above-described algorithm runs in linear time. Now, we claim that the described algorithm computes the maximum possible tournament value achievable by any seeding.

It is straightforward to see that if, after the last iteration, the described algorithm computes a maximum current tournament value $V$, then there is a seeding that achieves tournament value $V$.

For the converse, let the maximum achievable tournament value be $V$, and assume for contradiction that the maximum current tournament value computed by our algorithm after the last iteration is strictly smaller than $V$. We say that two seedings $\sigma$ and $\sigma'$ \emph{agree on a prefix of size $\ell$} of the players if each of the $\ell$ strongest players has the same number of wins in $\sigma$ and $\sigma'$. 
Let $\sigma$ be a seeding achieving tournament value $V$ that agrees on a maximum prefix of the players with the seeding produced by the described algorithm.
Let $i$ be the strongest player that is not part of the prefix.
We make the following case distinction:

    \smallskip
    \noindent \textbf{Case 1.} Player $i$ is popular. The algorithm seeds player $i$ as a winner of a subtournament with $r$ rounds but $\sigma$ seeds players $i$ as a winner of a subtournament with $r'\neq r$ rounds. First, note that we must have that $r'<r$, since all players that can beat $i$ are seeded by the algorithm in the same way as by $\sigma$, and the algorithm seeds player $i$ as the winner of a subtournament with a maximum number of rounds.

    Now, let player $j$ be the strongest player with $i>j$ that is seeded by $\sigma$ as a winner of a subtournament with $r$ rounds. Note that such a player must exist. Next, we consider two subcases.
    \begin{enumerate}

    \item Assume $j$ is an unpopular player, then we swap the seed positions of players $i$ and $j$ in $\sigma$. This results in a new seeding $\sigma'$ where player $i$ is the winner of a subtournament with $r$ rounds. Furthermore, the tournament value achieved by $\sigma'$ is at least $V$: player $j$ cannot win more than $r'$ games, since player $i$ is stronger than player $j$. It follows that the games won by player $i$ when the seeding $\sigma$ is used are now won by either player $j$ or some other player. In particular, the $r'$ games that player $i$ won when $\sigma$ is used give a value of at least $r'\cdot b$ when $\sigma'$ is used. Furthermore, the $r$ games won by player $j$ in $\sigma$ are now won by player $i$ in $\sigma'$ and hence give a value of $r\cdot a$. Let $V'$ denote the tournament value obtained when the seeding $\sigma'$ is used. We have $$V'\ge V+r\cdot (a-b)+r'\cdot (b-a)=V+(r-r')\cdot(a-b)\ge V.$$ 
    
    Furthermore, the seeding $\sigma'$ agrees on a prefix with the seeding produced by the described algorithm that contains at least one more player (namely player $i$). This is a contradiction to the assumption that $\sigma$ is a seeding achieving tournament value $V$ that agrees on a maximum prefix of the players with the seeding produced by the described algorithm.

  \item  Assume $j$ is a popular player. In this case, unlike the previous case, we need to swap the seed positions of multiple players as follows: Informally, we swap the seed positions of player $i$ and all players seeded in the subtournament rooted by the $r'$th game of player $i$ with the seed positions player $j$ and all players seeded in the subtournament rooted by the $r'$th game of player $j$. Formally, we rearrange the players in the seeding $\sigma$ by swapping the players on positions $\sigma(i), \ldots, \sigma(i)+2^{r'}-1$ in $\sigma$ with players on positions $\sigma(j), \ldots, \sigma(j)+2^{r'}-1 $ in $\sigma$ (maintaining the order). See \cref{fig:8} for an illustration of this swap. 
    This results in a new seeding $\sigma'$ where, in particular, player $i$ is the winner of a subtournament with $r$ rounds. Furthermore, player~$j$ wins at least $r'$ games in $\sigma'$, since player $j$ beats the same players in the first $r'$ rounds in $\sigma$ and $\sigma'$. Hence, we have that the tournament value obtained by $\sigma'$ is at least $V$. Furthermore, the seeding $\sigma'$ agrees on a prefix with the seeding produced by the described algorithm that contains at least one more player (namely player~$i$). This is a contradiction to the assumption that $\sigma$ is a seeding achieving tournament value $V$ that agrees on a maximum prefix of the players with the seeding produced by the described algorithm.
    \end{enumerate}
    
    \smallskip
    \noindent \textbf{Case 2.} Player $i$ is unpopular. The algorithm seeds player $i$ as a winner of a subtournament with $r$ rounds but $\sigma$ seeds players $i$ as a winner of a subtournament with $r'\neq r$ rounds. First, note that we must have that $r'>r$, since all players that can beat $i$ are seeded by the algorithm in the same way as by $\sigma$, and the algorithm seeds player $i$ as the winner of a subtournament with a minimum number of rounds.

     \begin{figure}[t]
 \centering
    \includegraphics[scale=0.65]{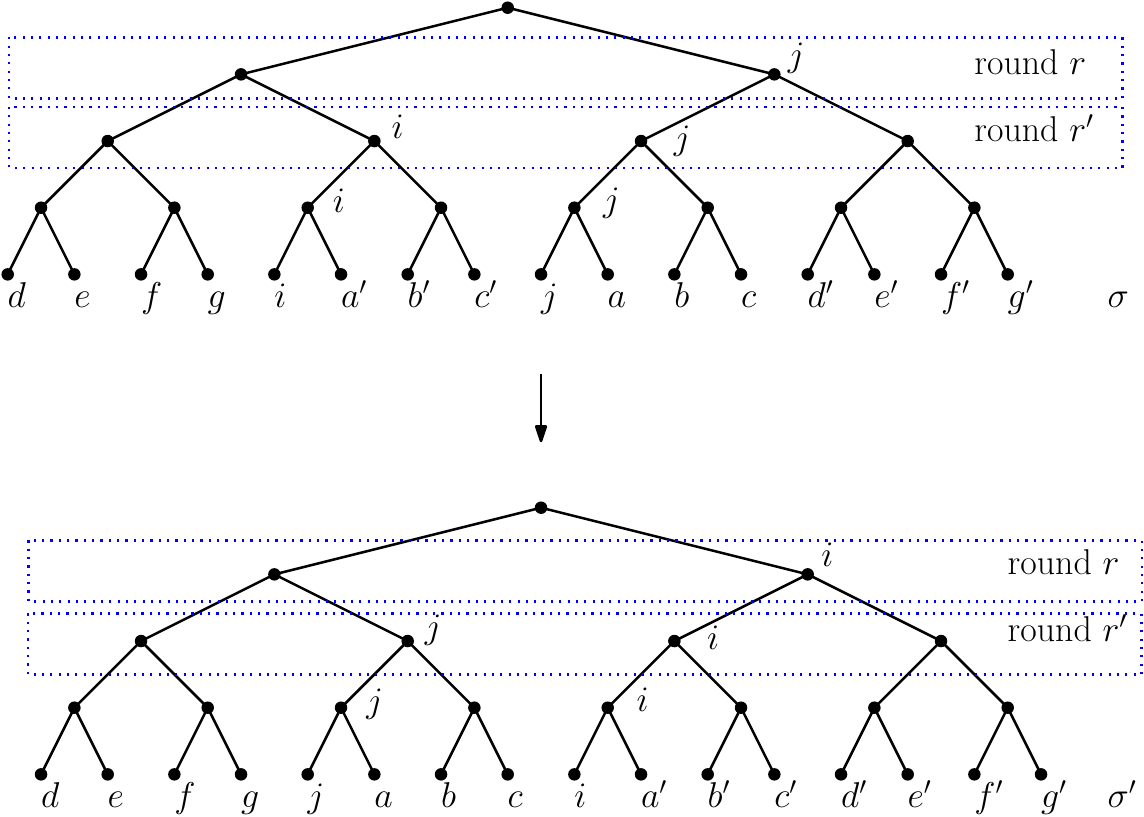}
    \caption{An illustration of the swapping of the players on positions $\sigma(i), \ldots, \sigma(i)+2^{r'}-1$ with players on positions $\sigma(j), \ldots, \sigma(j)+2^{r'}-1. $ Here, note that $r'=2$ and $r=3$. Also, for clarity, only a part of the whole tournament is shown.}
    \label{fig:8}
\end{figure}


    Now, let player $j$ be the strongest player with $i>j$ that is seeded as a winner of a subtournament with $r$ rounds. Note that such a player must exist. From here, the analysis is analogous to the first case, except that popular and unpopular players switch their roles.

It follows that the described algorithm is correct.
\end{proof}

Next, we present an FPT algorithm for the case where the strength ordering of the players and the player popularity value ordering agree for most players. 
Formally, we say that the player popularity values \emph{agree} with the strength ordering if for all $i,j\in N$ we have that if $i>j$ then $v_i\ge v_j$, where $v_i$ and $v_j$ are the player popularity values of $i$ and $j$, respectively. 
 We say that the \emph{disagreement} between the player popularity values and the strength ordering is $k$ if $k$ is the smallest integer such that there exists a player set $N'\subseteq N$ with $|N'|\le k$ and for all $i,j\in N\setminus N'$ we have that if $i>j$ then $v_i\ge v_j$, where $v_i$ and $v_j$ are the player popularity values of $i$ and $j$, respectively. We call $N'$ a minimum set of disagreeing players.

We use the disagreement between the player popularity values and the strength ordering as a parameter.
Intuitively, we can guess for the disagreeing players how many wins they should get, and then we can use the greedy algorithm from \cref{thm:twovalues} and treat the remaining players as if they were popular.

First, we show that we can compute a minimum set of disagreeing players efficiently, since we are going to need access to this set in our FPT algorithm.

\begin{proposition}\label{prop:disagree}
Given an instance of \probname\ with a player popularity-based game-value function, we can compute a minimum set of disagreeing players for that instance in $\OO(k2^k\cdot n+n\log n)$ time, where $k$ is the disagreement between the player popularity values and the strength ordering.  
\end{proposition}

\begin{proof}
For every player $i\in N$, let $v_i$ denote the player popularity value of player $i$.
    First, we sort the players lexicographically by their player popularity value and their strength. Let $>_\text{value}$ denote this ordering, and let $>_\text{strength}$ denote the strength ordering. We first show that we can find a pair $i,j\in N$ of players with $i>j$ and $v_i<v_j$ in linear time (if one exists). 

    If there exist $i,j\in N$ of players with $i>j$ and $v_i<v_j$, then the two orderings $>_\text{strength}$ and $>_\text{value}$ are not the same, since $i >_\text{strength} j$ but $j >_\text{value} i$. Then, there must exist an ordinal position $\ell$ such that player $i_s$ is at ordinal position $\ell$ of $>_\text{strength}$ and player $i_v$ is at ordinal position $\ell$ of $>_\text{value}$ and $i_s\neq i_v$. Assume that $\ell$ is the smallest such ordinal position. Then, we have that $i_s >_\text{strength} i_v$ since otherwise $i_v$ must have an ordinal position $\ell'<\ell$ in $>_\text{strength}$, a contradiction to the assumption that $\ell$ is the smallest ordinal position where $>_\text{strength}$ and $>_\text{value}$ do not have the same player.
    
    Furthermore, we have $i_v >_\text{value} i_s$ since otherwise $i_s$ must have an ordinal position $\ell'<\ell$ in $>_\text{value}$, a contradiction to the assumption that $\ell$ is the smallest ordinal position where $>_\text{strength}$ and $>_\text{value}$ do not have the same player. Lastly, we have that $v_i<v_j$ since if $v_i=v_j$, then $i_s >_\text{value} i_v$ because of the lexicographical ordering. 

    It follows that we can iterate through the ordinal positions of $>_\text{strength}$ and $>_\text{strength}$ and compare the two respective players. If the two players are different, then we have found a pair $i,j\in N$ of players with $i>j$ and $v_i<v_j$. Clearly, this takes $\OO(n)$ time.

    Now, to compute a set of disagreeing players of size $k$, we can find a pair of players $i,j\in N$ of players with $i>j$ and $v_i<v_j$ and branch on which of the two players we put into the set of disagreeing players. Then, we recursively find a set of disagreeing players of size $k-1$ among the remaining players. Note that this resembles the well-known basic branching algorithm for \textsc{Vertex Cover} parameterized by the solution size~\cite{CyganFKLMPPS15}. We avoid explicitly constructing a graph on the set of players (as vertices) with edges between pairs of players $i,j\in N$ of players with $i>j$ and $v_i<v_j$ to keep the polynomial part of the running time in $\OO(n)$ (after the sorting step). To find a minimum $k$, we start with $k=0$ and increase $k$ by one until we find a set of disagreeing players of size $k$.
    It is straightforward to see that the described algorithm has the claimed running time. Its correctness follows from the correctness of the well-known basic branching algorithm for \textsc{Vertex Cover} parameterized by the solution size~\cite{CyganFKLMPPS15}.
\end{proof}

Now, we are ready to present the FPT algorithm for the minimum disagreement as a parameter.


\begin{theorem}\label{thm:fptdisagree}
    \probname\ with a player popularity-based game-value function is solvable in $k^{\OO(k)}\cdot n^{1+o(1)}$ time, where $k$ is the disagreement between the player popularity values and the strength ordering.
\end{theorem}

\begin{proof}

First, we compute a minimum set $N'\subseteq N$ of disagreeing players using \cref{prop:disagree}.
%
For each player in $N'$, we guess how many games they win. This yields $(\log n)^{k}$ possibilities. 
Intuitively, once we guess the number of wins for each player in $N'$, we ``reserve'' subtournaments of appropriate size for the players in $N'$, and we use the greedy algorithm described in the proof of \cref{thm:twovalues} to fill up the remaining players in the tournament (by treating them as if they were popular).

More formally, the algorithm proceeds as follows. Once we guess the number of wins for each player in $N'$, we sort the players in $N'$ lexicographically by the number of games they are supposed to win and their strength. For every $i\in N$, let $v_i$ denote the player popularity value of player $i$. 
Initially, one subtournament with $\log n$ rounds is open, and the current tournament value is zero. We consider pairs of players in every iteration, where one player is from the set $N\setminus N'$ and one player is from the set $N'$. 
We iterate through the players in $N\setminus N'$ in order of their strength, starting with the strongest player, and we iterate through the players in $N'$ according to their sorting. Whenever we seed a player, we have to take care that we do not seed a stronger player into one of the subtournaments that will open after seeding that particular player. Therefore, whenever we open subtournaments as a result of seeding a player $i\in N$, we say that those subtournaments are \emph{restricted} by $i$.
Let players $i\in N\setminus N'$ and $j\in N'$ be the players considered in a current iteration. 
Let $r$ denote the largest number of rounds of an open subtournament. 
\begin{itemize}
    \item If we have guessed $j$ to have $r$ wins and there is a subtournament with $r$ rounds that is not restricted by some player $j'<j$, then we close the open tournament with $r$ rounds that is restricted by the weakest player $j''$ that is stronger than $j$, that is, $j''>j$, and open $r$ new subtournaments, for each round $r'\in[r-1]$ one with $r'$ rounds and, if $r\neq 0$, we open one degenerate subtournament for round zero. Each opened subtournament is restricted by $j$.
    Furthermore, we add $r\cdot v_j$ to the current tournament value. If there is no such subtournament, we abort. Otherwise, we proceed with considering player $i$ and the next player in $N'$.

\item If we have guessed $j$ to have strictly less than $r$ wins. Then, we seed player $i$ as a winner of an open subtournament with the largest number of rounds, say $r'$, that is not restricted by some player $j'<i$. Then, we close the open tournament with $r'$ rounds that is restricted by the weakest player $j''$ that is stronger than $i$, that is, $j''>i$, and open $r'$ new subtournaments, for each round $r''\in[r'-1]$ one with $r''$ rounds and, if $r'\neq 0$, we open one degenerate subtournament for round zero. Each opened subtournament is restricted by $i$. 
Furthermore, we add $r'\cdot v_i$ to the current tournament value. If there is no such subtournament, we abort. Otherwise, we proceed with considering the next player in $N\setminus N'$ and player $j$.
\end{itemize}

We first analyze the running time of the described algorithm. First, note that by \cref{prop:disagree}, we can compute a minimum set of disagreeing players in $\OO(k2^k\cdot n+n\log n)$ time. Given a set of disagreeing players, we have $(\log n)^{k}$ possibilities on how many games each of these players should win (called ``guesses''). Given a guess, we iterate through all players. Whenever we seed a player, we open and close restricted subtournaments. Note that in total, there are $\OO(n)$ games played, and for each one, we open and close exactly one restricted subtournament. Organizing these subtournaments in e.g.\ a balanced binary search tree~\cite{knuth2015art}, we can open and close a restricted subtournament in $\OO(\log n)$ time. Furthermore, we can find an appropriate subtournament for a given player in $\OO(\log n)$ time. It follows that the overall running time is in
\[
\OO(k2^k\cdot n+n\log n) + \OO((\log n)^{k}\cdot n\log n)\subseteq (\log n)^{\OO(k)}\cdot n.
\]
It is well-known that if $k\le n$, then $(\log n)^{\OO(k)}\subseteq k^{\OO(k)}+n^{\OO(1)}$ (see e.g.~\cite{RamanujanS17}). However, we can give a tighter analysis: Consider the case that $k<\log n / (\log \log n)^2$. Then, we have that 
\[
(\log n)^{\OO(k)}=n^{c/(\log \log n)},
\]
for some constant $c$. On the other hand, if $k\ge\log n / (\log \log n)^2$ then we have that 
\[
(\log n)^{\OO(k)}\le (k\cdot (\log \log n)^2)^{\OO(k)}.
\]
For large enough $n$ we further have that $k>(\log \log n)^2$, hence we get that 
\[
(\log n)^{\OO(k)} \subseteq k^{\OO(k)}.
\]
It follows that $(\log n)^{\OO(k)}\subseteq k^{\OO(k)}+n^{o(1)}$ and hence, we obtain the claimed running time bound.

In the remainder, we argue that the described algorithm computes the maximum possible tournament value achievable by any seeding.

It is straightforward to see that if after the last iteration for some guess, the described algorithm computes a maximum current tournament value $V$, then there is a seeding that achieves tournament value $V$.

For the converse, let the maximum achievable tournament value be $V$ and let $\sigma$ be a seeding achieving tournament value $V$. Consider an iteration of the algorithm where the number of guessed wins for the players in $N'$ is the same as when seeding $\sigma$ is used. 
Note that the algorithm does not abort in this case: for every size of open subtounaments, the players in $N'$ are seeded into those subtournaments first. Since seeding a player in $N$ into some open subtournament only opens up new subtournaments of smaller size, we have a maximal amount of open subtournaments of each size once we reach the iteration where they are considered. Hence, it does not happen that we cannot find a subtournament for them.

Assume for contradiction that the maximum current tournament value computed by the algorithm in that iteration is strictly smaller than $V$. We refine the notion of two seedings agreeing on a prefix of players that we used in the proof of \cref{thm:twovalues}. For this proof, we say that a seeding $\sigma$ and the seeding produced by the algorithm \emph{agree on a prefix of size $\ell$} if each of the first $\ell$ players seeded by the algorithm has the same number of wins in $\sigma$ the seeding produced by the algorithm \emph{and} is beaten by the same player in $\sigma$ and the seeding produced by the algorithm. 

Now, let $\sigma'$ be a seeding achieving tournament value $V$, where the number of wins for the players in $N'$ is the same as when seeding $\sigma$ is used, and additionally such that $\sigma'$ and the seeding produced by the algorithm agree on a maximum prefix. 
Let $i$ be the first player that is not part of the prefix.
We make a case distinction on whether $i$ is a player in $N'$ or not.

\begin{itemize}
    \item Consider the case that $i\in N'$. Then, we know by the assumption that player $i$ has the same number of wins, say $r$, in $\sigma'$ and the seeding produced by the algorithm. It follows that the player that wins against $i$ must be different in $\sigma'$ and the seeding produced by the algorithm. Let player $j$ beat $i$ when $\sigma'$ is used, and let player $j'$ beat $i$ in the seeding produced by the algorithm. Furthermore, let $i'$ be the player that is beaten by $j'$ in round $r+1$ when $\sigma'$ is used.
    
    First, we argue that $i'$ must be weaker than $i$: If $i'$ was stronger than $i$. Then, $i$ is seeded in an earlier iteration by the algorithm than $i$. However, $i'$ is beaten by different players in $\sigma'$ and the seeding produced by the algorithm. This is a contradiction to the assumption that $i$ is the first player that is seeded differently in $\sigma'$ and the algorithm. Hence, we can conclude that~$i>i'$.

     Similar as in the second case in the proof of \cref{thm:twovalues}, we swap the seed positions in $\sigma'$ of multiple players as follows: Informally, we swap the seed positions of player $i$, and all players seeded in the subtournament rooted by the $r$th game of player $i$ with the seed positions player $i'$ and all players seeded in the subtournament rooted by the $r$th game of player $i$. Formally, we rearrange the players in the seeding $\sigma'$ by swapping the players on positions $\sigma'(i), \ldots, \sigma'(i)+2^{r}-1$ in $\sigma$ with players on positions $\sigma'(i'), \ldots, \sigma'(i')+2^{r}-1 $ in $\sigma'$ (maintaining the order). 

     Note that since $j>i>i'$ and $j'>i>i'$, we have that after the swapping, every player wins exactly the same number of games, and only players $i$ and $i'$ are beaten by different players. Hence the tournament value does not change. However, after the swapping, player $i$ is beaten by the same player as in the seeding produced by the algorithm. This is a contradiction to the assumption that $\sigma'$ is a seeding with tournament value $V$ that agrees on a maximum prefix with the seeding produced by the algorithm.
    \item Consider the case that $i\in N\setminus N'$. Now, we consider two subcases: One where player $i$ has the same number of wins in $\sigma'$ and the seeding produced by the algorithm but is beaten by a different player, and one where player $i$ has fewer wins in $\sigma'$ than in the seeding produced by the algorithm.

    For the first case, we can make an argument analogously to the case above where $i\in N$. 

    For the second case, let the algorithm seed player $i$ as a winner of a subtournament with $r$ rounds but $\sigma'$ seed players $i$ as a winner of a subtournament with $r'\neq r$ rounds. First, note that we must have that $r'<r$, since the algorithm first seeds players as winners of larger subtournaments.
Let player $j$ beat player $i$ in the seeding produced by the algorithm. Now, let player $i'\in N$ be the player that is seeded as a winner of a subtournament with $r$ rounds and beaten by $j$ in $\sigma'$. Note that such a player must exist and note that we must have $j\in N\setminus N'$ since players in $N'$ that should obtain $r$ wins are seeded first, that is, if $j\in N'$ then $j$ is seeded before player $i$ by the algorithm and hence must be beaten by the same player in $\sigma'$ and the seeding produced by the algorithm.

 Similar to the case where $i\in N$ and the second case in the proof of \cref{thm:twovalues}, we swap the seed positions of multiple players as follows: Informally, we swap the seed positions of player $i$ and all players seeded in the subtournament rooted by the $r'$th game of player $i$ with the seed positions player $i'$ and all players seeded in the subtournament rooted by the $r'$th game of player $i'$. Formally, we rearrange the players in the seeding $\sigma'$ by swapping the players on positions $\sigma'(i), \ldots, \sigma'(i)+2^{r'}-1$ in $\sigma'$ with players on positions $\sigma'(i'), \ldots, \sigma'(i')+2^{r'}-1 $ in $\sigma'$ (maintaining the order). See \cref{fig:8} for an illustration of this swap. 
    This results in a new seeding $\sigma''$ where, in particular, player $i$ is the winner of a subtournament with $r$ rounds and is beaten by player $j$. Furthermore, player~$i'$ wins at least $r'$ games in $\sigma'$, since player $i'$ beats the same players in the first $r'$ rounds in $\sigma'$ and $\sigma''$. Hence, we have that the tournament value obtained by $\sigma''$ is at least $V$. Furthermore, the seeding $\sigma''$ agrees on a prefix with the seeding produced by the described algorithm that contains at least one more player (namely player~$i$). This is a contradiction to the assumption that $\sigma'$ is a seeding achieving tournament value $V$ that agrees on a maximum prefix of the players with the seeding produced by the described algorithm.

%
\end{itemize}

    It follows that the described algorithm is correct.
\end{proof}

We point out that if $k=0$, that is, the set of disagreeing players is empty, then we can solve the problem in linear time (rather than $n^{1+o(1)}$ time). Note that we can identify this case in linear time, since in the proof of \cref{prop:disagree} we show that we can find a pair of players where the strength disagrees with the player popularity values in linear time. Furthermore, the guessing step is unnecessary in this case. Finally, we do not need to keep track by which players the subtournaments are restricted, hence there are no $\log n$ factors in the running time. Hence, we get the following corollary.
\begin{corollary}\label{thm:sameorder}
    \probname\ with a player popularity-based game-value function is solvable in linear time if the player popularity values agree with the strength ordering
\end{corollary}
 We further remark that the corollary also follows from results obtained by \citet{dagaev2018competitive}.
Finally, we remark that \probname\ is presumably not $\mathsf{NP}$-hard in the restricted setting, which our FPT algorithm is able to solve, since we have a quasi-polynomial time algorithm for this case (\cref{thm:quasipoly}).

\subsection{An FPT-Algorithm for Parameter Size of Influential Set of Players}\label{sec:fptvc}

In this section, we study the parameterized complexity of \probname\ when there exists a small \emph{influential set of players}. 
\begin{definition}
Given an instance of \probname, a set $F\subseteq N$ of players is \emph{influential} if for all $i,j\in N$ and $r\in\mathbb{N}$ with $v(i,j,r)\neq 0$, we have that $i\in F$ or $j\in F$.
\end{definition}
First, we observe that we can compute a minimum-sized influential set of players efficiently, since we can easily reduce the problem to finding a minimum vertex cover in an appropriately defined graph.
\begin{observation}
A minimum influential set $F$ for an instance $I$ of \probname\ can be computed in $1.2738^{|F|}\cdot |I|^{\mathcal{O}(1)}$ time.
\end{observation}
This follows from the fact that an influential set of players is a vertex cover of the graph that has the set of players $N$ as vertices and an edge between two players $i,j\in N$ if there exists an $r\in\mathbb{N}$ such that $v(i,j,r)\neq 0$. \textsc{Vertex Cover} is fixed-parameter tractable when parameterized by the solution size, and a minimum vertex cover can be computed in the claimed running time~\cite{chen2010improved}.

Formally, we show the following.
\begin{theorem}\label{thm:fptinfluence}
    \probname\ with round-oblivious game-value function is fixed-parameter tractable when parameterized by the size of a minimum influential set of players.
\end{theorem}

To show \cref{thm:fptinfluence}, we adapt an algorithm for {\sc Tournament Fixing} parameterized by the feedback vertex number by \citet{zehavi2023tournament}. We point out that due to the similarity of the algorithms, we need to introduce many concepts of \citet{zehavi2023tournament} to present our algorithm. We will use the same notation and terminology as \citet{zehavi2023tournament}.

In the {\sc Tournament Fixing} problem, we are given a set of players and a tournament graph that defines the outcome of a game between each pair of players. We are asked whether there exists a seeding that makes a specific player win the tournament. The parameter is the feedback vertex number of said tournament graph. Note that once the players corresponding to the feedback vertex set are removed, the remaining players admit a linear ordering defining their strength, similar to our setting. 
Informally, the main strategy of the algorithm of \citet{zehavi2023tournament} is as follows. 
\begin{itemize}
    \item ``Guess'' a sufficient amount of information on how the players corresponding to the feedback vertex set perform in the tournament.
    \item Find suitable opponents for the players corresponding to the feedback vertex set. Recall that those players do not fit into the linear ordering defining the strength. For each of those players, there are some games that they are supposed to win. For those games, we have to select players as opponents that they are able to beat. Similarly, for each of those players, there is at most one game that they are supposed to lose. For that game, we have to select a player as an opponent that beats them.
    \item Make sure that we can fill up the rest of the tournament with the remaining players.
\end{itemize}
We can solve our problem by following a very similar approach.
\begin{itemize}
    \item We want to ``guess'' a sufficient amount of information on how the players in the influential set of players perform in the tournament.
     \item We want to find suitable opponents for the players in the influential set of players, since those games are the only ones giving us value.
    \item We want to make sure that we can fill up the rest of the tournament with the remaining players.
\end{itemize}

To present the algorithm more formally, we adopt an alternative view of tournament seeds that gives a better way of looking at the performance of a player in a tournament. 
Given a tournament seeding $\sigma$ for a set $N$ of players, we define a \emph{tournament execution tree} $T^{(\sigma)}$.
\begin{definition}[Tournament Execution Tree]
Given a tournament seeding $\sigma$ for a set $N$, the rooted tree $T^{(\sigma)}$ is a \emph{tournament execution tree} if the following holds.
\begin{itemize}
    \item The vertices of $T^{(\sigma)}$ are the players in $N$, that is, $V(T^{(\sigma)})=N$.
    \item There is an edge between two players $i,j\in N$ in $T^{(\sigma)}$ if player $i$ plays against player $j$ when the seeding $\sigma$ is used.
    \item The tree $T^{(\sigma)}$ is rooted at the winner $\max N$ of the tournament.
\end{itemize}    
\end{definition}
We give an illustration of a tournament tree in \cref{fig:10}.

\begin{figure}[t]
 \centering
    \includegraphics[scale=0.7]{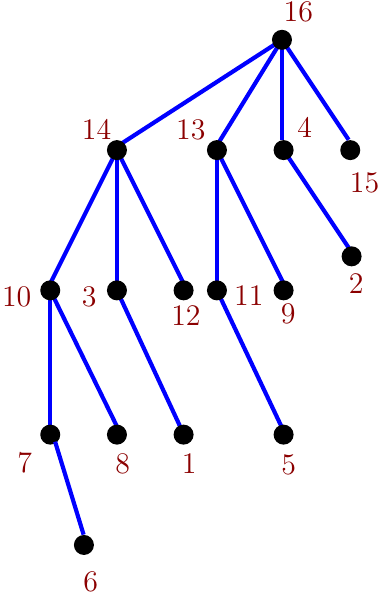}
    \caption{Example of a tournament execution tree with 16 players $\{1,2,\ldots,16\}$ and the seeding $(16,15,4,2,13,9,11,5,7,6,10,8,12,14,3,1)$, where the position of a player in the vector corresponds to its seed position.}
    \label{fig:10}
\end{figure}

 We use the following notation. Given a rooted tree $T$ and a vertex $u\in V(T)$, let $\mathsf{parent}_T(u)$ denote the parent of $u$ in $T$ (for vertices $u$ that are different from the root of $T$) and let $\mathsf{children}_T(u)$ denote the set of children of $u$ in $T$.
 It is easy to observe that for $i\in N$, we have that $\mathsf{children}_{T^{(\sigma)}}(i)$ are the players knocked out of the tournament by player $i$ when seeding $\sigma$ is used. Since the vertices of a tournament execution tree are the players of the tournament, we sometimes refer to the vertices of a tournament execution tree as players.

A Tournament execution tree has a specific structure that is captured by the concept of a \emph{binomial arborescence}. Formally, a binomial arborescence is defined as follows.

\begin{definition}[Binomial Arborescence (BA)]\label{def:BA}
A {\em binomial arborescence (BA)} $B$ rooted at $u$ is defined recursively:
\begin{itemize}
\item The tree on the single vertex $u$ is a BA rooted at $u$.
\item Given two vertex-disjoint BAs of equal size, $B_u$ rooted at $u$ and $B_{u'}$ rooted at $u'$,  the addition of an edge from $u$ to $u'$ yields a BA rooted at $u$.
\end{itemize}
For any $n$ that is a power of $2$, the unique (up to isomorphism) BA on $n$ vertices is denoted by $B_n$.
\end{definition}

We immediately get the following observation.

\begin{observation}\label{obs:seedfromBA}
Let $\sigma$ be a tournament seeding for $N$ players with $|N|=2^{n'}$ for some $n'\in\mathbb{N}$. Then, $T^{(\sigma)}$ is isomorphic to $B_{|N|}$.

Let $B$ be a BA with $V(B)=N$. If for all $i\in N$ and $j\in \mathsf{children}_B(i)$ it holds that $i>j$, then there exists a tournament seeding $\sigma$ such that $T^{(\sigma)}=B$. Furthermore, $\sigma$ can be computed in polynomial time.
\end{observation}

Binomial arborescences have many nice properties that we will exploit in our algorithm. In particular, we will make use of the following immediate observation regarding BAs.

\begin{observation}\label{obs:BAprops}
Let $B=B_n$ for some $n$ that is a power of $2$ and let $B'$ denote the subtree of $B$ rooted at some $u\in V(B)$.
Then, we have that
    $B'$ is a BA,
    $|\mathsf{children}_B(u)|=\log_2|V(B')|$, and
    the length of the path in $B$ from the root of $B$ to $u$ is at most $\log_2 n$.
\end{observation}

For the remainder of the section, fix some instance $I$ of \probname, let $F$ be a minimum influential set of players of $I$, and let $w=\max N$ be the strongest player. Let $F^\star=F\cup\{w\}$. We use $S$ to denote a tournament execution tree $T^{(\sigma^\star)}$ for a tournament seeding~$\sigma^\star$ that maximizes the tournament value that our algorithm will attempt to find.

The main strategy of the algorithm is to (bijectively) map all players in $N$ to the vertices of the BA $B_{|N|}$ in order to obtain a tournament execution tree. By \cref{obs:seedfromBA}, this implicitly defines a tournament seeding $\sigma$, which can be computed from the mapping in polynomial time. The goal is to find a mapping such that the tournament value is maximized, that is, the mapping produces $S$ from $B_{|N|}$. To this end, the algorithm performs the following steps.
\begin{enumerate}
    \item ``Guess'' a minimum amount of information about the ``positions'' of the vertices in $F^\star$ in $S$.
    \item Use a dynamic program to find a mapping from $N$ to the vertices of $B_{|N|}$ that respects the guessed ``positions'' and produces $S$ from $B_{|N|}$.
\end{enumerate}

\paragraph*{Step 1: Guesses.}
First, we would like to have some information on the ancestral relationship between the vertices in $S$ corresponding to players in $F^\star$. We model this information using a so-called \emph{structure}. To give the formal definition of a structure, we first introduce {\em lowest common ancestors} and {\em least common ancestor-closures (LCA-closures)}.

Given a rooted tree $T$ and two vertices $u,u'\in V(T)$, the {\em lowest common ancestor} of $u$ and $u'$ in $T$ is denoted by $\mathsf{lca}_T(u,u')$.

\begin{definition}[LCA-closure]
For a rooted tree $T$ and a subset $X\subseteq V(T)$, the {\em least common ancestor-closure (LCA-closure)}, $\mathsf{LCA}_T(X)$, is obtained by the following process. Initially, set $X' = X$. Then, as long as there exist vertices $u,u'\in X'$ such that $\mathsf{lca}_T(u,u')\notin X'$, add $\mathsf{lca}_T(u,u')$ to $X'$. When the process terminates, output $\mathsf{LCA}_T(X)=X'$.
\end{definition}


In the structure we are about to define, some vertices are the players in $F^\star$, while others are ``placeholders'' to describe the positions of $\mathsf{LCA}_S(F^\star)\setminus F^\star$ (for the unknown $S$ that the algorithm tries to find) as well as some of the other vertices adjacent to the vertices in $F^\star$ in $S$.

\begin{definition}[Structure]\label{def:structure}
A {\em structure} is a tree $T$ rooted at $w$ such that $V(T)\cap N=F^\star$, and both of the following conditions hold:
\begin{enumerate}
\item\label{def:structure1} For every $u,u'\in F^\star$ with $u'\in \mathsf{children}_T(u)$ we have $u>u'$.
\item\label{def:structure2} For every $u\in V(T)\setminus N$, at least one of the following conditions holds:
\begin{enumerate}
\item\label{def:structure2a} $\mathsf{children}_T(u)\cap F\neq\emptyset$.
\item\label{def:structure2b} $\mathsf{parent}_T(u)\in F^\star$ and $\mathsf{children}_T(u)\neq\emptyset$.
\item\label{def:structure2c} $|\mathsf{children}_T(u)|\geq 2$.
\end{enumerate}
\end{enumerate}
\end{definition}

For a structure to encode the aforementioned ancestral relationship we will add some ``annotations'' to it, which we formally define soon. 
The second condition is meant to ensure that we do not add too many placeholders to the structure: we have placeholders for parents of vertices in $F^\star$, their children that have descendants in $F^\star$, and vertices corresponding to the LCA-closure; so, in total, we have $\OO(|F^\star|)$ placeholders. This is required to obtain the running time bound of the overall algorithm.

The mentioned annotations we add to a structure are the following: 
\begin{itemize}
    \item For each of its edges $\{u,u'\}$, we would like to know the length of the path from the vertex to which $u$ is mapped and the vertex to which $u'$ is mapped in $S$.
    \item For each of its vertices $u$, we would like to know the size of the subtree rooted at the vertex to which $u$ is mapped in $S$.
\end{itemize}
Knowing the above quantities allows us to calculate how many vertices we need to ``fill'' in these paths and subtrees. 
In particular, we want to know how large the subtrees in $B$ that are rooted at some vertex corresponding to a placeholder in the structure are. Furthermore, we want to know the length of a path in $B$ between two vertices corresponding to adjacent placeholders in the structure. To this end, we use a \emph{length function} and a \emph{size function}. Note that due to \cref{obs:BAprops}, we know that the lengths of the paths are upper-bounded by $\log_2 |N|$.

\begin{definition}[Length Function]\label{def:length}
Given a structure $T$, a function $\mathsf{len}: E(T)\rightarrow\{1,2,\ldots,\log_2|N|\}$ is a  {\em length function} if for every edge $\{u,u'\}\in E(T)$ such that $\{u,u'\}\cap F^\star\neq\emptyset$, $\mathsf{len}(u,u')=0$.
\end{definition}

In particular, the definition above ensures that we have placeholders for all the neighbors (in $B$) of the vertices in $F^\star$. From the definition of a binomial arborescence (\cref{def:BA}) we know that the sizes of subtrees have a very specific pattern.

\begin{definition}[Size Function]\label{def:size}
Given a structure $T$, a function $\mathsf{siz}: V(T)\rightarrow\{2^0,2^1,\ldots,2^{\log_2|N|}\}$ is a  {\em size function}.
\end{definition}

A triple $(T,\mathsf{len},\mathsf{siz})$, consisting of a structure, a length function, and a size function, respectively, is called an {\em annotated structure}. Next, we define what it means for an annotated structure to be \emph{realizable}. Intuitively, this is the case if we can map the vertices of the structure to a BA $B$ in a way that ``complies'' with the length and the size function.

\begin{definition}[Realizability]\label{def:guess}
Let $(T,\mathsf{len},\mathsf{siz})$  be an annotated structure. An injective function $\mathsf{rea}: V(T)\rightarrow V(B)$ is a {\em realizability witness} for $(T,\mathsf{len},\mathsf{siz})$ if all of the following conditions hold:
\begin{enumerate}
\item\label{def:guess1} For every two vertices $u,u'\in V(T)$, $\mathsf{rea}(\mathsf{lca}_T(u,u'))=\mathsf{lca}_{B}(\mathsf{rea}(u),\mathsf{rea}(u'))$. In particular, $\mathsf{LCA}_B(\{\mathsf{rea}(u): u\in V(T)\})=\{\mathsf{rea}(u): u\in V(T)\}$.
\item\label{def:guess2} For every edge $e=\{u,u'\}\in E(T)$, the length of the path from $\mathsf{rea}(u)$ to $\mathsf{rea}(u')$ in $B$ is $\mathsf{len}(e)$.
\item\label{def:guess3} For every vertex $u\in V(T)$, the size of the subtree of $B$ rooted at $\mathsf{rea}(u)$ is $\mathsf{siz}(u)$.
\end{enumerate}
If $(T,\mathsf{len},\mathsf{siz})$ admits a realizability witness, then it is said to be {\em realizable}.
\end{definition}

The first condition ensures that we respect ancestral relationships encoded by the structure. The other two conditions ensure that we respect the length and size functions.
We give an illustration of a realizable structure in \cref{fig:11}.

\begin{figure}[t]
 \centering
    \includegraphics[scale=0.7]{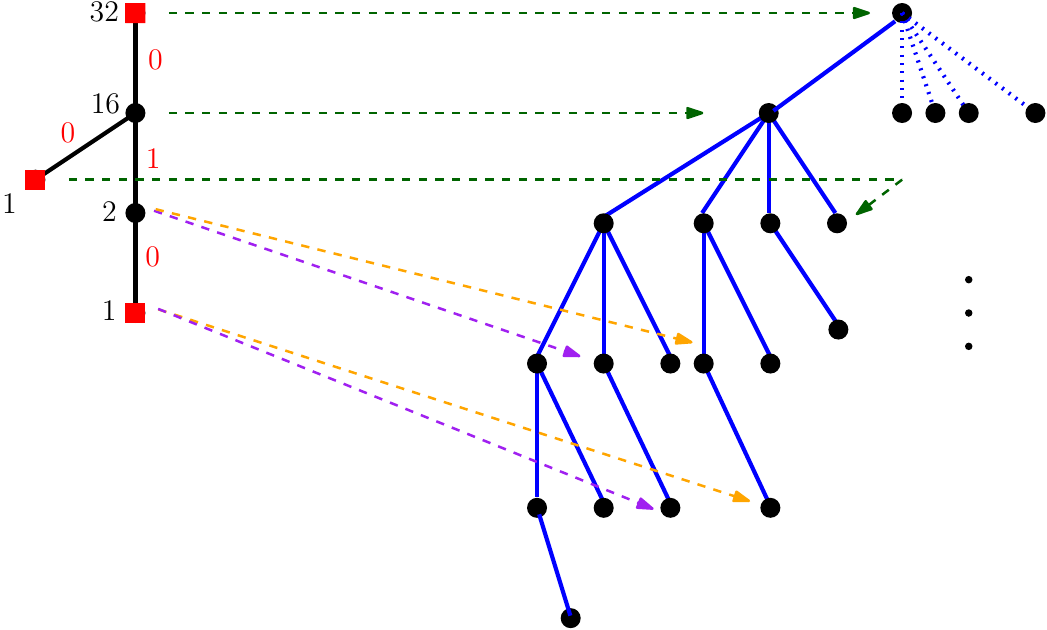}
    \caption{Example of a realizable annotated structure. Here, we have $n=32$ and $|F^\star|=3$. The annotated structure $(T,\mathsf{len},\mathsf{siz})$ is depicted on the left, the binomial arborescence $B$ is (partially) depicted on the right. In the annotated structure, vertices in $F^\star$ are depicted as red squares and vertices in $T\setminus F^\star$ as black circles. The size function is written in black and the length function is written in red. The mapping $\mathsf{rea}$ is visualized by the dashed arrows. For the top three vertices of the annotated structure, there is only one possible mapping, visualized by the greed dashed arrows. For the bottom two vertices of the annotated structure, there are two possible mappings that can be used by $\mathsf{rea}$, visualized by the violet and orange dashed arrows, respectively.}
    \label{fig:11}
\end{figure}

Now, we can give the definition of the notion of \emph{compliance} of an annotated structure and a tournament execution tree. Intuitively, if an annotated structure complies with some tournament execution tree, then it correctly encodes the information about that tournament execution tree.

\begin{definition}[Compliance]\label{def:compliance}
An annotated structure $(T,\mathsf{len},\mathsf{siz})$ {\em complies} with a tournament execution tree $S$ if $(T,\mathsf{len},\mathsf{siz})$ has a realizability witness $\mathsf{rea}: V(T)\rightarrow V(S)$ such that: for every $u\in F^\star$, $\mathsf{rea}(u)=u$.
\end{definition}



\citet{zehavi2023tournament} showed that we can efficiently compute a sufficiently small collection of annotated structures such that at least one of them complies with $S$.

\begin{lemma}[\cite{zehavi2023tournament}]\label{lem:guesses}
There exists a $|F^\star|^{\OO(|F^\star|)}\cdot n^{\OO(1)}$-time algorithm that outputs a collection $\cal C$ of $|F^\star|^{\OO(|F^\star|)}\cdot n^{\OO(1)}$ realizable annotated structures such that $\cal C$ contains an annotated structure that complies with $S$.
\end{lemma}

\paragraph*{Step 2: Dynamic Programming.}


We will refer to the dynamic programming table as $M$. From now on we fix some realizable annotated structure $(T,\mathsf{len},\mathsf{siz})$. Each entry of our table corresponds to a so-called {\em index}, defined below.

\begin{definition}[Index]\label{def:index}
An {\em index} is a quadruple $I=(i,f,g,h)$ where $i\in\{0,1,\ldots,|N|-|F^\star|\}$, $f: V(T)\setminus F^\star\rightarrow\{\mathsf{false},\mathsf{true}\}$, $g: E(T)\rightarrow \{0,1,\ldots,\log_2|N|-1\}$ and $h: F^\star\rightarrow \{0,1,\ldots,\log_2|N|\}$, such that all of the following conditions hold:
\begin{enumerate}
\item\label{def:index1} For every $u\in V(T)\setminus F^\star$ such that $f(u)=\mathsf{true}$, either $\mathsf{parent}_T(u)\in F^\star$ or both $f(\mathsf{parent}_T(u))=\mathsf{true}$ and $g((\mathsf{parent}_T(u),u))=\mathsf{len}(e)-1$.
\item\label{def:index2} For every $e\in E(T)$, $g(e)\leq \mathsf{len}(e)-1$.
\item\label{def:index3} For every $e=\{u,u'\}\in E(T)$ such that $u=\mathsf{parent}_T(u')$ and $g(e)>0$, ($u\notin F^\star$ and)\footnote{Follows from \cref{def:length} and Condition \ref{def:index2} in \cref{def:index}.} $f(u)=\mathsf{true}$.
\item\label{def:index4} For every $u\in F^\star$, $h(u)\leq \log_2\mathsf{siz}(u)-|\mathsf{children}_T(u)|$.
\end{enumerate} 
\end{definition}

Intuitively speaking, an entry $M[(i,f,g,h)]$ will the maximum value of games played in a ``partial solution'' (that is, mapping of the players to the vertices in $B$) involving the players in $F^\star$ and the $i$ strongest players that are not in $F^\star$ that, in particular, satisfies the following: 
\begin{itemize}
    \item $f$ encodes, for $u\in V(T)\setminus F^\star$, whether the partial solution maps a player from $N$ to the vertex corresponding to $u$ in $B$ (according to some realizability witness);
    \item $g$ encodes, for $\{u,u'\}\in E(T)$, to how many internal vertices on the path from the vertex corresponding to $u$ to the vertex corresponding to $u'$ the partial solution maps players from~$N$;
    \item $h$ encodes, for $u\in F^\star$, the number of children of the vertex corresponding to $u$ in $B$---excluding those corresponding to children of $u$ in $T$---to whom the partial solution maps players from~$N$.
\end{itemize}
 Note that Conditions \ref{def:index2} and \ref{def:index4} ensure that we respect the length and size functions.

We do not provide a formal correctness proof of the dynamic program, because its correctness follows from similar arguments as made by \citet{zehavi2023tournament}. However, we give an informal description of the main ideas here. The correctness of the dynamic program heavily relies on two notions:
\begin{enumerate}
    \item \emph{validity} of indices and
    \item \emph{witnesses} for indices.
\end{enumerate}
We do not give the technical definitions here but describe on an intuitive level what these concepts capture. We first discuss \emph{valid indices}, since they are important for describing the recursive formula for the dynamic program. After we give the recursive formula, we discuss \emph{index witnesses}.

We are going to map the players in $N\setminus F^\star$ to the vertices of $B$ in order of their strength. Intuitively, we want to map players to certain vertices in $B$ before players are mapped to certain other vertices in $B$. In the definition of indices (\cref{def:index}), some of the constraints aim to enforce such precedences.  
Condition~\ref{def:index1} ensures that if we map a player from $N\setminus F^\star$ to the vertex corresponding to some vertex $u\in V(T)\setminus F^\star$, then we have already mapped players from $N$ to the vertex corresponding to its parent in $T$ as well as to all the vertices on the path between them.
Condition~\ref{def:index3} ensures that if we map a player from $N\setminus F^\star$ to the path in $B$ between two adjacent vertices $u,u'\in V(T)\setminus F^\star$ with $u=\mathsf{parent}_T(u')$, then we already mapped a player from $N\setminus F^\star$ to the vertex in $B$ corresponding to $u$.

Furthermore, some players are implicitly set aside to ``fill up'' subtrees of $B$ rooted at e.g.\ children of a vertex in $F^\star$ or children of vertices on a path between two vertices in $F^\star$. Note that we can only start to fill up subtrees once a player is mapped to its root vertex in $B$. Considering paths between two vertices in $F^\star$: we can assume w.l.o.g.\ that first players are mapped to \emph{all} vertices along the path, and then we start to fill up the subtrees of those vertices.
If an index is valid, those subtrees must be large enough to host all vertices designated to fill up subtrees, since the player located at the root of the subtree must be strong enough to beat all other players in the subtree.

\citet{zehavi2023tournament} formalized the notion of valid indices and showed that they can be identified efficiently. This involves some technically involved arguments that we omit here.

Finally, we are ready to define the table $M$. The table $M$ maps indices to natural numbers (and~$-\infty$).
%
%
%
%
%
The order in which we map players to $B$ corresponds to $\sigma$. Concretely, when we consider an index $(i,f,g,h)$, the set of players that we map is the set of the first $i$ players according to $\sigma$ that are not contained in $F^\star$. 
The basis is when $i=0$. We have two cases:
\begin{enumerate}
\item If $f$ is the function that assigns only $\mathsf{false}$, $g$ is the function that assigns only $0$, and $h$ is the function that assigns only $0$, then $M[I]$ stores
\[
\sum_{\{u,u'\}\in E(T)\wedge \{u,u'\}\subseteq F^\star} v(u,u').
\]
\item Otherwise, $M[I]$ stores $-\infty$.
\end{enumerate}

Now, suppose that $i\geq 1$. 
If $I$ is an invalid index, then $M[I]$ stores $-\infty$.
Otherwise, $M[I]$ stores the maximum of the following entries. Let $p_i$ denote the $i$th strongest player (according to $\sigma$) that is not contained in $F^\star$.
\begin{enumerate}
\item\label{item1} $M[(i-1,f,g,h)]$.
\item\label{item2} $M[(i-1,f',g,h)]+v^\star$ for any $f'$ and $v^\star$ such that:
	\begin{enumerate}
	\item $f'$ is identical to $f$ except for exactly one vertex $u^\star\in V(T)\setminus F^\star$ to which $f$ assigns $\mathsf{true}$ but $f'$ assigns $\mathsf{false}$,
	\item for every $u\in\mathsf{children}_T(u^\star)$ with $u\in F^\star$: $p_i>u$ (according to $\sigma$),
	\item for $u=\mathsf{parent}_T(u^\star)$: if $u\in F^\star$, then $u>p_i$ (according to $\sigma$), and
	\item
 \[
 v^\star=\sum_{u\in\mathsf{children}_T(u^\star)\cap F^\star}v(p_i,u) + \sum_{u\in\{\mathsf{parent}_T(u^\star)\}\cap F^\star}v(p_i,u).
 \]
 \end{enumerate}
\item\label{item3} $M[(i-1,f,g',h)]$ for any $g'$ that is identical to $g$ except for exactly one edge $e^\star$ to which $g$ assigns some number $\ell\geq 1$ but $g'$ assigns $\ell-1$.
\item\label{item4} $M[(i-1,f,g,h')]+v^\star$ for any $h'$ and $v^\star$ such that:
	\begin{enumerate}
	\item $h'$ is identical to $h$ except for exactly one vertex $u^\star\in F^\star$ to which $h$ assigns some number $\ell\geq 1$ but $h'$ assigns $\ell-1$,
	\item $u^\star>p_i$ (according to $\sigma$), and
    \item $v^\star=v(u^\star,p_i)$.
	\end{enumerate}
\end{enumerate}

Intuitively, the four cases correspond to different types of ``locations'' in $B$ to which $p_i$ is mapped. In the first case, it is mapped to a vertex that has no neighbors nor descendants corresponding to players in $F^\star$. In the second case, it is mapped to a vertex corresponding to a vertex in $T$. In the third case, it is mapped to an internal vertex on the path between two vertices corresponding to vertices in $T$.  Lastly, in the fourth case, it is mapped to the child corresponding to a player in $F^\star$.



The time complexity of the dynamic program is straightforward to obtain and is analogous to the analysis done by \citet{zehavi2023tournament}.

\begin{lemma}[\cite{zehavi2023tournament}]\label{lem:runtimeM}
Table $M$ can be computed in time $|F^\star|^{\OO(|F^\star|)}\cdot n^{\OO(1)}$.
\end{lemma}

As mentioned earlier, the correctness of the dynamic program relies on the notion of index witnesses. Informally speaking, an index witness for some valid index $(i,f,h,g)$ is a mapping of the players in $F^\star$ and the $i$ strongest players that are not in $F^\star$ to the vertices of $B$ such that
\begin{itemize}
    \item The mapping respects the constraints implicitly given by $f$, $h$, and $g$.
    \item There exists a realizability witness for the annotated structure fixed in the beginning that agrees with the mapping of the index witness.
\end{itemize}

Note that the mapping of the index witness implicitly defines some games that are played (between players that are mapped by the index witness to adjacent vertices in $B$). Informally, the dynamic program is correct if and only if for all valid indices $I$, there exists an index witness such that the values of the games implicitly defined by the mapping of the index witness add up to $M[I]$. A formal proof of this statement is analogous to the one given by \citet{zehavi2023tournament}, which is technically quite involved. We omit the formal correctness proof here.

\paragraph*{The Final Algorithm.}





Given an instance \probname, the algorithm proceeds as follows: 
\begin{enumerate} 
\item Call the algorithm in \cref{lem:guesses} to compute a collection $\cal C$ of realizable annotated structures. Set $V_{\max}=0$.
\item For every $(T,\mathsf{len},\mathsf{siz})\in{\cal C}$, compute the table $M$ corresponding to $(T,\mathsf{len},\mathsf{siz})$ (see the Dynamic Programming section). If $M[I]>V_{\max}$ where $I$ is the index of $(T,\mathsf{len},\mathsf{siz})$, then set $V_{\max}=M[I]$.
\item Output Yes if $V_{\max}\ge V$, and No otherwise.	
\end{enumerate}

We first analyze the time complexity. By \cref{lem:guesses}, both $|{\cal C}|$ and the time necessary to compute ${\cal C}$ are upper-bounded by $|F^\star|^{\OO(|F^\star|)}\cdot n^{\OO(1)}$. Then, by \cref{lem:runtimeM}, we have that for every $(T,\mathsf{len},\mathsf{siz})\in{\cal C}$, the algorithm uses $|F^\star|^{\OO(|F^\star|)}\cdot n^{\OO(1)}$ time. So the overall running time is upper-bounded by $|F^\star|^{\OO(|F^\star|)}\cdot n^{\OO(1)}$.

We do not provide a formal correctness proof. The full proof is analogous to the correctness proof by \citet{zehavi2023tournament}. However, roughly speaking, the correctness is implied by the following arguments.
\begin{enumerate}
    \item The mapping from the sought tournament execution tree $S$ to $B$ is an index witness for indices corresponding to annotated structures that comply with $S$.
    \item The mapping of an index witness of an index corresponding to a realizable annotated structure maps the players to $B$ such that $B$ is a tournament execution tree.
    \item If we face a yes instance, then there is a realizable annotated structure that complies with a tournament execution tree $S$ (that gives the maximum tournament value). By the first argument, the mapping from $S$ to $B$ is an index witness for indices corresponding to the annotated structure. By \cref{lem:guesses}, there is an iteration where said annotated structure is considered, and by the correctness of the dynamic program, the algorithm returns Yes.
    \item If the algorithm outputs Yes, then by the second argument, the mapping of the index witness corresponding to the index that maximized the tournament value defines a tournament execution tree. By the correctness of the dynamic program, that tournament execution tree produces the desired tournament value.
\end{enumerate}

\section{Conclusion and Open Problems} \label{sec:conc}

In this paper, we formulated an objective of finding seedings for a knockout tournament that maximizes the profit or popularity of the tournaments. We provided a spectrum of results, showcasing both negative and (mostly) positive results.

Still, several intriguing open questions remain.
\begin{itemize}
\item Can we improve the $(1/\log n)$-approximation bound presented in \cref{apx:algo}, or can we give stronger approximation hardness results? Additionally, is there an approximation algorithm for \probname~that can handle non-round oblivious game-value functions?
\item Given that we have a quasipolynomial-time algorithm for \probname\ with a win-count oriented game-value function, it is natural to ask the following: Is it possible to solve \probname\ for win-count oriented game-value functions in polynomial time, or can we establish a conditional lower bound for the problem?
\item What is the computational complexity of \probname~when the game-value functions are player popularity-based, and there are more than two distinct player popularity values? 
\item We have restricted the game-value functions in our paper to be round oblivious extensively, as our primary aim was to lay a theoretical groundwork to fully understand and investigate the simpler game-value functions first. 
It seems very difficult, algorithmically, to handle game values that can arbitrarily change in different rounds unless all possibilities are checked, such as in our algorithm for win-count-oriented game value functions (\cref{thm:quasipoly}), which can be non-round-oblivious. However, given the limitations of round obliviousness in real-world contexts, are there natural restrictions for non-round-oblivious game-value functions that still allow for additional algorithmic results?  
\item 
What is the computational complexity of \probname~when other natural restrictions on game-value functions are considered, e.g.\ when players have popularity values and the value of a game is defined as the product of the popularity values of the two players playing the game?

\item Lastly, can we replace the parameter ``size of the influential set of players'' with a natural and smaller structural parameter and obtain tractability?

\end{itemize}

As the field continues to evolve, our research might inspire further exploration and advancements in optimizing tournament seeding strategies from the viewpoint of the organizers.

\bibliographystyle{abbrvnat}
\bibliography{bibliography}

\end{document}